\documentclass[pra,superscriptaddress]{revtex4}
\usepackage[dvips]{graphicx}
\usepackage{graphicx,epsfig}
\usepackage{amsmath,amsfonts,amssymb,amsthm,amsbsy,mathtools}
\usepackage{bbold}

\newcommand{\ket}[1]{\mathop{\big|#1\big>}\nolimits}            

\newcommand{\kbr}[2]{| #1\rangle\!\langle #2 |}
\newcommand{\Tr}[1]{\mathop{{\mathrm{Tr}}_{#1}}}
\newcommand{\diag}{\mathop{{\mathrm{diag}}}}

\newcommand{\sldc}{{\mathsf{sl}(d,\mathbb{C})}}
\newcommand{\sltw}{\mathsf{sl}(2,\mathbb{C})}
\newcommand{\slthr}{\mathsf{sl}(3,\mathbb{C})}
\newcommand{\slfo}{\mathsf{sl}(4,\mathbb{C})}

\newtheorem{thm}{Theorem}

\newtheorem{lem}[thm]{Lemma}
\newtheorem{defi}[thm]{Definition}

\theoremstyle{remark}
\newtheorem*{rem}{Remark}
\newtheorem*{exa}{\bf Example}
\newtheorem{cor}[thm]{\bf Corollary}

\def\openone{\mathbb{1}}
\def\bbC{\mathbb{C}}

\def\bbR{\mathbb{R}}

\newcommand{\nn}{\nonumber}
\def\dg{\dagger}

\def\vcm{\vec m}
\def\vcn{\vec n}

\def\a{\alpha}
\def\b{\beta}
\def\g{\gamma}
\def\d{\delta}

\def\ve{\varepsilon}
\def\vr{\varrho}

\def\s{\sigma}

\def\la{\lambda}
\def\La{\Lambda}

\def\B{\mathcal{B}}
\def\G{\mathcal{G}}
\def\O{\mathcal{O}}
\def\D{\mathcal{D}}
\def\K{\mathcal{K}}
\def\H{\mathcal{H}}
\def\M{\mathcal{M}}
\def\T{\mathcal{T}}
\def\U{\mathcal{U}}

\begin{document}

\title{Capacities of Grassmann channels}

\author{Kamil Br\'adler}
 \email{kbradler@cs.mcgill.ca}
 \affiliation{
    School of Computer Science,
    McGill University,
    Montreal, Quebec, H3A 2A7, Canada
    }
\author{Tomas Jochym-O'Connor}
 \affiliation{
    School of Computer Science,
    McGill University,
    Montreal, Quebec, H3A 2A7, Canada
    }
\affiliation{
    Institute for Quantum Computing and Department of Physics and Astronomy,
    University of Waterloo,
    200 University Avenue West, Waterloo, ON, N2L 3G1, Canada
    }
\author{Roc\'io J\'auregui}
\affiliation{
    Instituto de F\'isica, Universidad Nacional Aut\'onoma de M\'exico, Apdo. Postal 20-364, M\'exico D.F. 01000
    }

\date{April 28, 2011}

\begin{abstract}
A new class of quantum channels called Grassmann channels is introduced and their classical and quantum capacity is calculated. The channel class appears in a study of the two-mode squeezing operator constructed from operators satisfying the fermionic algebra. We compare Grassmann channels with the channels induced by the bosonic two-mode squeezing operator. Among other results, we challenge the relevance of calculating entanglement measures to assess or compare the ability of bosonic and fermionic states to send quantum information to uniformly accelerated frames.
\end{abstract}

\maketitle

\section{Introduction}
The notion of quantum channel capacity is central to quantum Shannon theory. Early development in the seventies~\cite{early} was a starting point to an impressive amount of knowledge that has been acquired in the last two decades~\cite{todo,era_priv}. Two of the most investigated areas is the classical~\cite{classcap} and quantum capacity~\cite{LSD} of a quantum channel. The classical/quantum capacity informs us about the ability of a quantum channel to transmit classical or quantum correlations. More precisely, consider a sender who has, in principle, at his disposal an optimal encoder producing a classical or quantum code, and a receiver able to process the channel output and recover the transmitted information (that is, to decode) with arbitrarily high precision. In this way, the information can be transmitted at the rate given by the capacity and cannot be improved by any other choice of encoding. Various additional conditions or restrictions might be added, for instance if privacy is required~\cite{privacy} or some sort of assistance in terms of other quantum or classical resources available to the communicating parties~\cite{sym_ass,ent_ass}. This leads to a large number of important capacity definitions relevant under given circumstances and one might even try to characterize the whole capacity multi-dimensional regions in which the axes correspond to various available resources~\cite{capregion,capregion_had,capregion_qudits}.

Due to the presence of regularization (see below) the classical or quantum capacity is not efficiently computable. There are, however, particular examples of channels for which the classical or quantum capacity is easy to calculate. In the case of the classical capacity, every such example must be cherished since the proof usually involves some nontrivial manipulations~\cite{classcap_additivity,classcap_additivity_clone,hadamard,entbreak,TransDepol,WolfEisert}. For the quantum capacity, almost all known non-trivial examples fall in the class of degradable channels~\cite{degchan,erasure_channel,giovannetti}. Among these examples are exceptional cases for which both capacities are known, to our knowledge there are only two examples: (i) a qubit erasure channel~\cite{erasure_channel} and (ii) Hadamard channels~\cite{hadamard}. There are also trivial examples of such channels with zero quantum capacity: entanglement-breaking channels~\cite{entbreak} and anti-degradable channels. In this paper we add another member into the elite group of non-trivial examples: the Grassmann channels. In order to find the quantum capacity we show that the Grassmann channels are degradable while to find the classical capacity we make use of the fact that the Grassmann channels are of a direct sum form. We show that the channels which form the structure of the Grassmann channels are new members of the surprisingly broad family of channels studied in~\cite{WolfEisert} for which the classical capacity is efficiently calculable. The lowest-dimensional example of the Grassmann channels turns out to be a qubit erasure channel. In general, higher-dimensional Grassmann channels have certain traits in common with a qubit erasure channel making them interesting in the light of some recent capacity results~\cite{era_ass,era_priv}.

This paper can be partially seen as an accompanying paper to Ref.~\cite{CMP}. There, one of the authors investigates a channel and its capacities induced by the bosonic squeezing transformation, playing the
role of a channel isometry. The isometry is presented from the physical point of view as it appears in the context of the Unruh effect for massless scalar fields. The channel was nicknamed the Unruh channel and its properties and analysis from the quantum Shannon theory point of view were also presented in other works~\cite{JHEP,capregion_had,classcap_additivity_clone,ConjDeg}. The current paper is the `fermionic' version of the analysis done in the bosonic case, in the sense that the isometry is generated by operators obeying the canonical anticommutation relations (the fermionic algebra). The consequences for quantum Shannon theory, mentioned in the previous paragraph, are radically different from the bosonic case which justifies an in-depth study of this fermionic case. For example, we show that the Grassmann channels do not belong to the class of Hadamard channels, as opposed to the bosonic case~\cite{capregion_had}. In terms of physical consequences, the capacity calculations for the Grassmann channels are even more surprising. In contrary to common opinion, we show that: (i) it makes absolutely no qualitative difference for the study of the Unruh effect whether the states are composed of bosons or fermions (at least for quantum information transmission purposes) and (ii) naive calculations of certain entanglement measures of states shared between inertial and uniformly accelerated observers do not provide much insight into the channel's ability to reliably send quantum information. Finally, there is an intriguing connection of the Grassmann channels and their bosonic relatives to the well studied family of transpose-depolarizing channels.

The encounter with fermionic degrees of freedom brings an interesting complication. Due to the use of fermionic statistics, the usual objects studied in quantum information theory, such as entangled or separable states, need to be treated carefully.  As we discuss later in more detail, the problem lies in the fact that many-particle fermionic systems lack the tensor product structure. This is one of the simplest examples of the braided statistics and has relatively recently begun to be studied more closely by a number of authors~\cite{fermions}. Due to the presence of tensor products in all capacity definitions, it is not immediately clear how to generalize this concept to a fermionic system, or even a system with more general statistics. We will show that in our case, however, we may use the `standard' framework for qubits (qudits) due to a careful choice of the input encoding for fermionic states and the specific isometry that we investigate in this work. Throughout this paper we use the multi-rail encoding and we will show that under given circumstances the capacity formulas indeed need not be modified. Thus, we may proceed and practice quantum Shannon theory with no modifications.

The connection of quantum capacities and Grassmann variables has previously been briefly visited also in a different context by the authors of~\cite{grass_cap} elaborating on the formalism introduced in~\cite{cahglaub}.

The paper is organized as follows. In Sec.~\ref{sec:def} we recall some basic notions from quantum Shannon theory together with the fermionic algebra and briefly discuss various relevant physical situations. Sec.~\ref{sec:grassconstruction} describes the construction of the qudit Grassmann channels and study some of their properties. Sections~\ref{sec:quantcap} and~\ref{sec:classcap} contain the calculation of the quantum and classical capacity of the whole (infinite-dimensional) class of Grassmann channels. In Sec.~\ref{sec:consequences} we discuss several physical properties of the Grassmann channels.

\section{Definitions and notation}\label{sec:def}

\subsection*{Classical and quantum capacity}
Let $\B(\H)$ be the algebra of bounded operators on a Hilbert space $\H$ which for our purposes will be a $d$-dimensional complex vector space $\H=\bbC^d$. Our algebra is therefore just the algebra of complex $d$-dimensional matrices. Let $\K:\B(\H_{A'})\mapsto\B(\H_{A'})$ be a quantum channel mapping density operators from an input Hilbert space $\H_{A'}$ to an output Hilbert space $\H_A$ (if there is no chance of confusion we will use the shorthand $\K:A'\mapsto A$). We define $V_\K:A'\mapsto AC$ to be the channel's isometric extension. Equivalently, by embedding the isometry into a higher-dimensional Hilbert space we write $U_\K:A'C'\mapsto AC$, where $U_\K$ is a unitary operator. Using the isometry picture the complementary channel of $\K$ is defined as $\K^c(\s)=\Tr{A}[V_\K\s V^\dg_\K]\equiv\Tr{A}\circ\ V_\K\circ\s$. The second equality is how we will occasionally abbreviate similar expressions. When there is no chance of confusion we will omit the mode index $A$ or $C$ for states. The von Neumann entropy $H(\vr_A)=-\Tr{}[\vr_A\log{\vr_A}]$ will be written in the economic way as $H(A)_\vr$ which is mainly suitable for dealing with the parts of multipartite states $\vr_A=\Tr{BC\dots}[\vr_{ABC\dots}]$. This convention will also be used for other entropic quantities.

A classical ensemble  can be written in the form of a classical-quantum state $\s_{XA'}=\sum_xp(x)\kbr{x}{x}_X\otimes\s_{x,A'}$ with a probability distribution function $p(x)$. Then, for a quantum channel $\K$ the capacity is given by the regularized expression~\cite{classcap}
\begin{equation}\label{eq:classcap}
C(\K)=\lim_{n\to\infty}{1\over n}{C_{\rm Hol}(\K^{\otimes n})}.
\end{equation}
$C_{\rm Hol}(\K)$ is the Holevo quantity
\begin{equation}\label{eq:holcap}
C_{\rm Hol}(\K)=\sup_{\{p(x),\s_{x,A'}\}}{I(X:A)}_{\varsigma},
\end{equation}
where $\varsigma_{XA}=\K(\s_{XA'})$ and $I(X:A)=H(X)+H(A)-H(XA)$ is the mutual information. For the optimization task in Eq.~(\ref{eq:holcap}) it is sufficient to consider $\s_{x,A'}$ to be pure states.

The expression for the quantum capacity contains regularization as well~\cite{LSD}
\begin{equation}\label{eq:quantcap}
Q(\K)=\lim_{n\to\infty}{{1\over n}Q^{(1)}(\K^{\otimes n})}.
\end{equation}
$Q^{(1)}(\K)$ is the optimized coherent information
\begin{equation}\label{eq:optimcoh}
Q^{(1)}(\K)=\sup_{\psi}{[H(A)_\tau-H(C)_\tau}],
\end{equation}
where $\tau_{AC}=V_\K\circ\psi_{A'}$.

There exists a generalization of the von Neumann entropy known as the $\a$-R\'enyi entropy $H^\a(\rho)\overset{\rm df}{=}1/(1-\a)\log{[\Tr{}\rho^\a]}$, where the von Neumann entropy is recovered for $\a\to1$. A quantum channel $\D$ is called degradable~\cite{degchan,degchan_struc} if there exists another channel $\M$ such that $\M\circ\D=\D^c$ holds. The channel $\M$ is called a degrading map. Finally, $\ln$ is the natural logarithm and $\log$ denotes the base $d$ logarithm unless stated otherwise. The motivation for using the base $d$ logarithm is a better way of comparing the capacities for the class of channels studied here and the their `bosonic' equivalent studied elsewhere~\cite{CMP}.

\subsection*{Fermions and their correlated pairs}

For a set of modes that are specified by quantum numbers compactly denoted by $\kappa$, the exchange characteristics
of indistinguishable fermions that may occupy those modes are reflected in the anticommutation relations obeyed
by the corresponding creation (annihilation) operators $a_\kappa$  ($a_\kappa^\dg$):
\begin{align}\label{eq:anticomm}
\{a_\kappa,a^\dg_{\kappa^\prime}\}&=\d(\kappa-\kappa^\prime),\\
\{a^\dg_\kappa,a^\dg_{\kappa^\prime}\}&=\{a_\kappa,a_{\kappa^\prime}\}=0.
\end{align}
An interesting transformation that preserves the anticommutation relations results from the following algorithm.
Within the available modes two subsets are chosen. Each of the
modes from the subset of lower cardinality is related through an
injective function $f(\kappa)=\kappa^\prime$ to the elements of the other subset.
In this way, the modes are paired and so are the creation-annihilation operators
$a_\kappa, a^\dg_\kappa$ and $a_{f(\kappa)}, a^\dg_{f(\kappa)}$. In order to make simpler
the notation, the operators $a_{f(\kappa)}, a^\dg_{f(\kappa)}$ from now on will
be denoted by a different letter, for instance, $a_{f(\kappa)}, a^\dg_{f(\kappa)}\rightarrow
c_\kappa,c^\dg_\kappa$. It can be directly shown that the so called Bogoliubov transformations
\begin{subequations}
\begin{align}
b_\kappa &= \cos{r}a_\kappa - e^{-i\phi}\sin{r}c^\dg_\kappa\label{eq:bogo1a}\\
b_\kappa^{\dg} &= \cos{r}a^\dg_\kappa -e^{i\phi} \sin{r}c_\kappa\label{eq:bogo1b}\\
d_\kappa &= \cos{r}c_\kappa + e^{-i\phi}\sin{r}a^\dg_\kappa\label{eq:bogo1c}\\
d_\kappa^{\dg} &= \cos{r}c^\dg_\kappa + e^{i\phi}\sin{r}a_\kappa.\label{eq:bogo1d}
\end{align}
\end{subequations}
preserve the anticommutation relations Eq.~(\ref{eq:anticomm}).
These transformations were first introduced for the generation of the Bardeen-Cooper-Schreiffer (BCS)
states, which are an excellent approximation to the ground
state of a weakly interacting superconductor~\cite{Schrieffer}.
In that case, the modes are thought to describe electrons with the quantum numbers
usually taken as the vector wave number $\vec k$ and the projection
of the spin $s$. The function $f$ is such that $f(\vec k,s) = (-\vec k,-s)$.
The Cooper pairs that are described using the operators $b_\kappa$, $d_\kappa$, as a consequence,
have opposite momenta and spin.

Bogoliubov transformations have been very useful in the
description of strongly correlated fermions in diverse scenarios like
condensed matter~\cite{condmat}, ultracold degenerate atomic Fermi gases~\cite{atoms}
or quantum field theory of particles on spacetime with a nontrivial metric~\cite{rocio,curved,unruh-fermi,unruh-schwarz}.
To illustrate the impact of these transformations in the latter area, consider negligible interacting massive fermions.  For a uniformly accelerated observer, the adequate spacetime coordinates are the Rindler ones;
for which two wedges, that can not be causally connected, are identified, we shall denote them by right $R$ and left $L$ wedges.
The right (left) Rindler modes are only supported on the right (left) wedge. An important question concerns to the energy spectra seen by a Rindler observer in connection to the Unruh effect~\cite{unruh76}.
The most clear treatments of the problem rely on comparisons between the so-called Unruh  modes and the Rindler modes.
They are naturally carried out in terms of generalized Bogoliubov transformations~\cite{rocio} that take into account that each Unruh mode needs for its representation an infinite superposition of Rindler modes.
In this kind of scenario, perhaps the simplest realization of a Bogoliubov transformation corresponds to  the case of a spin one-half fermion moving in a one dimensional space. Then, the modes of each fermion can be described in terms of a Grassmann field and a basic transformation is
\begin{equation}\label{eq:bogo2}
\begin{pmatrix}
  a^{(R)}_{k,s} \\
  \bar a^{(L)^\dg}_{-k,-s} \\
\end{pmatrix}=
\begin{pmatrix}
  \cos{r} & -e^{-i\phi}\sin{r} \\
  e^{i\phi}\sin{r} & \cos{r} \\
\end{pmatrix}
\begin{pmatrix}
  b_{k,s} \\
  \bar b_{-k,-s}^\dg \\
\end{pmatrix}.
\end{equation}
where $a,a^\dg$ stand for particles and $\bar a,\bar a^\dg$ for antiparticles, $k$ denotes the wave number and $s$ its spin projection. The real parameter $r$ can be  chosen to depend on the acceleration $\a$ and the rest mass $m$ of the Rindler observer as $\tan{r} = e^{-{\pi c^2\over\a}\sqrt{k^2+(mc/\hbar)^2}}$ \cite{unruh-fermi}.  Note that Eq.~(\ref{eq:bogo2}) is Eqs.~(\ref{eq:bogo1a}) and (\ref{eq:bogo1d}) written in the matrix form and with the notation adjusted to the relevant physical situation.

\section{Grassmann channels}\label{sec:grassconstruction}

\begin{defi}\label{defi:fermstates}
    We define a $d$-mode fermionic state as
    $$
    \ket{F}=\prod_{i=1}^d(a_i^\dg)^{n_i}\ket{vac}=\ket{n_1\dots n_d}\equiv\ket{\vec n},
    $$
    where $n_i\in\{0,1\}$. Following the properties of the fermionic operators we see that (i) each mode is occupied by at most one particle and (ii) the state $\ket{F}$ is completely antisymmetric. For $0\leq k\leq d$ there is ${d\choose k}$ possible fermionic states for which $\sum_{i=1}^{d}n_i=k$ holds.
\end{defi}
Note that throughout this article all input Hilbert spaces are spanned by the multi-rail basis.
\begin{defi}\label{defi:multirail}
    The multi-rail basis is defined as the set of all states $\ket{F}$ for which $\sum_{i=1}^d n_i=1$. The basis spans a $d$-dimensional Hilbert space and any fermionic state written in this basis is said to utilize the multi-rail encoding.
\end{defi}
The operator related to a Bogoliubov transformation~Eq.~(\ref{eq:bogo2}) over the $i$-th fermionic modes $a$ and $c$ reads
\begin{equation}\label{eq:exponential}
U_{A_iC_i} = \exp{\big[r(a_i^\dg c_i^\dg e^{-i\phi}-c_ia_i e^{i\phi})\big]},
\end{equation}
where $r, \phi\in\bbR$. The operator exponent may be factorized according to the following theorem.
\begin{thm}[\cite{disentangle}]\label{thm:disent}
    Let $J_+,J_-$ and $J_3$ be operators satisfying the commutation relations
    \begin{equation}\label{eq:comrels}
        \begin{split}
        [J_3,J_\pm] &= \pm J_\pm\\
        [J_+,J_-] &= 2J_3.
        \end{split}
    \end{equation}
    Then the following identity holds
    \begin{equation}
        e^{\la_+J_+ + \la_3J_3+\la_-J_-} = e^{\La_+J_+}e^{{\ln{\La_3}} J_3} e^{\La_-J_-},
    \end{equation}
    where
    \begin{equation}
        \begin{split}
        \La_\pm &= \frac{2\la_\pm\sinh{f}}{2f\cosh{f} -\la_3\sinh f}\\
        \La_3   &= \left(\cosh f - {\la_3\over 2f}\sinh f\right)^{-2}\\
        f       &= \big((\la_3/2)^2 +\la_-\la_+\big)^{1/2}.
        \end{split}
    \end{equation}
\end{thm}
In our case we have
\begin{equation}
    \begin{split}
    J_+ &= a^\dg c^\dg\\
    J_- &= ca = -ac\\
    J_3 &= {1\over2}\left(a^\dg a + c^\dg c -1\right).
    \end{split}
\end{equation}
This choice satisfies the commutation relations~in~Eqs.~(\ref{eq:comrels}). We get from Eq.~(\ref{eq:exponential}) $\la_\pm=\pm re^{\mp i\phi}$ and $\la_3=0$ leading to
\begin{equation}\label{eq:unitarydecomposed}
U_{A_iC_i}= \cos{r}\exp{[a_i^\dg c_i^\dg e^{-i\phi}\tan{r}]}\exp{[(a_i^\dg a_i+c_i^\dg c_i)\ln{\cos^{-1}{r}}]} \exp{[-c_ia_ie^{i\phi}\tan{r}]}.
\end{equation}
We might safely set $\phi=0$ since we will later see that it has no relevance in this work.

We collect the majority of identities used in the course of the paper in the following lemma.
\begin{lem}
    Let $a$ and $c$ be operators obeying the canonical anticommutation relations~Eq.~(\ref{eq:anticomm}). Then the following identities hold
    \begin{subequations}
    \begin{align}
        [a_ic_i,a^\dg_jc^\dg_j] &= 2\d_{ij}J_3\label{eq:comma}\\
        [a^\dg_ia_i,a^\dg_jc^\dg_j] &=  \d_{ij}a^\dg_jc^\dg_j\label{eq:commb}\\
        [a^\dg_ia_i,a^\dg_j] &= \d_{ij}a^\dg_j\label{eq:comanticom}\\
        [a_j^\dg c_j^\dg,a_i^\dg] &= 0      \label{eq:sigmaZXa}\\
        \bigg[\prod_{j=1}^ka_{j}c_{j},a_i^\dg\bigg] &=0      \label{eq:sigmaZXe}\\
        \bigg[\prod_{j=1}^ka^\dg_{j}c^\dg_{j},c_i^\dg\bigg] &=0  \label{eq:sigmaZXf}\\
        (-)^{\Delta_{k-1}}\prod_{j=1}^ka_jc_j &= \prod_{j=1}^ka_j\prod_{j=1}^kc_j\label{eq:acorder}.
    \end{align}
    \end{subequations}
    The first equation is a generalization of the second row in Eq.~(\ref{eq:comrels}), Eq.~(\ref{eq:comanticom}) simplifies to $a_i^\dg a_ia_i^\dg=a_i^\dg$ for $i=j$, Eq.~(\ref{eq:sigmaZXe}) holds for $i\not=j$ and $\Delta_{k-1}=k(k-1)/2$.
    \begin{proof}
        Eqs.~(\ref{eq:comma})-(\ref{eq:sigmaZXf}) directly follow from~Eq.~(\ref{eq:anticomm}). In Eq.~(\ref{eq:acorder}) we move all $a_j$ on the left of all $c_j$. So all $a_j$ on the LHS of Eq.~(\ref{eq:acorder}) for even $j$ `jump over' an odd number of $c_j$ operators. Similarly every odd $a_j$ switches its position with an even number of $c_j$ operators. The total acquired phase is
        $$
        \prod_{j=1}^k(-)^{(k-j)}=(-)^{\sum_{j=1}^k(k-j)}=(-)^{k(k-1)/2}\equiv(-)^{\Delta_{k-1}}.
        $$
    \end{proof}
\end{lem}
\noindent Taking into account Eq.~(\ref{eq:acorder}) the action of $d$ copies of the fermionic unitary operator results in
\begin{equation}\label{eq:fermi unitary action}
    \ket{\Psi}_{AC}=\bigotimes_{i=1}^d U_{A_iC_i}\ket{vac}=\cos^d{r}\sum_{k=0}^{d}(-)^{\Delta_{k-1}}\tan^k{r}\sum_{n_1,\dots,n_{d}}^{d\choose k}\ket{\vcn}_A\ket{\vcn}_C.
\end{equation}

If we wanted to see how $U_{AC}=\bigotimes_{i=1}^d U_{A_iC_i}$ transforms a fermionic qudit written in the multi-rail basis
$$
\ket{\psi}_{A'C'}=\sum_{i=1}^{d}\b_ia_i^\dg\ket{vac}=\sum_{i=1}^{d}\b_i\ket{i}_{A'}\ket{vac}_{C'}
$$
we might just calculate $\ket{\Phi}_{AC}=U_{AC}\ket{\psi}_{A'C'}$. In order to simplify this complicated calculation we first observe 
\begin{equation}\label{eq:simplified_exp}
    U_{AC}=\cos^d{r}\exp{\Big[\tan{r}\sum_{i=1}^da_i^\dg c_i^\dg\Big]}\exp{\Big[\sum_{i=1}^d(a_i^\dg a_i+c_i^\dg c_i)\ln{\cos^{-1}{r}}\Big]} \exp{\Big[\tan{r}\sum_{i=1}^da_ic_i\Big]},
\end{equation}
where  Eqs.~(\ref{eq:comma}) and (\ref{eq:commb}) for $i\not=j$ were utilized. To proceed we make use of the main advantage of the multi-rail encoding. Due to the presence of $c_i$ annihilating the vacuum state we observe
$$
    \exp{\Big[\tan{r}\sum_{i=1}^da_ic_i\Big]}\Big(\sum_{i=1}^{d}\b_ia_i^\dg\Big)\ket{vac}
    =\sum_{i=1}^{d}\b_ia_i^\dg\ket{vac}.
$$
We may then write
\begin{equation}\label{eq:greatcommuter}
\begin{split}
       \ket{\Phi}_{AC}
       &=\cos^{d}{r}\exp{\Big[\tan{r}\sum_{i=1}^da_i^\dg c_i^\dg\Big]}\exp{\Big[\sum_{i=1}^d(a_i^\dg a_i+c_i^\dg c_i)\ln{\cos^{-1}{r}}\Big]}\Big(\sum_{i=1}^{d}\b_ia_i^\dg\Big)\ket{vac}\\
       &=\cos^{d-1}{r}\exp{\Big[\tan{r}\sum_{i=1}^da_i^\dg c_i^\dg\Big]}\Big(\sum_{i=1}^{d}\b_ia_i^\dg\Big)\ket{vac}\\
       &=\cos^{d-1}{r}\Big(\sum_{i=1}^{d}\b_ia_i^\dg\Big)\exp{\Big[\tan{r}\sum_{i=1}^da_i^\dg c_i^\dg\Big]}\ket{vac}.
\end{split}
\end{equation}
The second row follows from Eq.~(\ref{eq:comanticom}) (see why $i=j$ does not spoil the commutator). The last equality is possible due to Eq.~(\ref{eq:sigmaZXa}).

The action of the unitary $U_{AB}$ leads to
\begin{subequations}\label{eq:transformed qudit}
\begin{align}
    \ket{\Phi}_{AC}&=U_{AC}\sum_{i=1}^{d}\b_i\ket{i}_{A'}\ket{vac}_{C'}
    =\cos^{d-1}{r}\Big(\sum_{i=1}^{d}\b_ia_i^\dg\Big)\exp{\Big[\tan{r}\sum_{j=1}^da_j^\dg c_j^\dg\Big]}\ket{vac}\label{eq:trans_qudita}\\
    &=\cos^{d-1}{r}\Big(\sum_{i=1}^{d}\b_ia_i^\dg\Big)\sum_{k=1}^{d+1}
    \tan^{k-1}{r}(-)^{\Delta_{k-2}}\sum_{N_k}\ket{\dots n_j\dots}_A\ket{\dots n_j\dots}_C\label{eq:trans_quditb}\\
    &=\cos^{d-1}{r}\sum_{k=1}^{d}\tan^{k-1}{r}
    (-)^{\Delta_{k-2}}\Bigg[
    \sum_{N_k}\sum_{i\in I}^{d-k+1}(\pm)_{i}\b_i\ket{\dots n_j+1_i\dots}_A\ket{\dots n_j\dots}_C
    \Bigg]\label{eq:trans_quditc},
\end{align}
\end{subequations}
where we have defined the set $N_k = \{ (n_1, \dots, n_d) | \sum_{j=1}^d n_j = k-1 \}$ to be the sum over all possible states with $k-1$ fermions. The additional fermion (with respect to the $C$ subsystem) in Eq.~\eqref{eq:trans_quditc} occupies the $i$-th position. There is ${d\choose k-1}$ states $\ket{\dots n_j\dots}_C$ in the middle sum of Eq.~(\ref{eq:trans_quditc}) for a given $N_k$ ($k=1\dots d$). The rightmost sum of Eq.~(\ref{eq:trans_quditc}) sums over $i \in I$, for a fixed $N_k$, such that $I = \{ i\text{ } | \text{ }  n_i = 0 \}$ and thus contains $d-k+1$ terms, justifying the upper limit in the sum. Taking, for example, $\b_0=1$ we verify $\cos^{2(d-1)}{r}\sum_{k=1}^{d}\tan^{2(k-1)}{r}{d-1\choose k-1}=1$. The multiplicative sign $(\pm)_{i}$ in Eq.~(\ref{eq:trans_quditc}) is dependent on the fermionic state of the $A$ subsystem. It appears following the operator rules specific to the fermionic algebra. Let us recall some basic properties of the algebra
\begin{equation}\label{eq:fermialgebra}
    \begin{split}
      a_i^\dg\ket{\dots n_j\dots} & = (1-n_i)(-)^{\sum_{j=1}^{i-1}n_j}\ket{\dots n_j+1_i\dots} \\
      a_i    \ket{\dots n_j\dots} & = n_i(-)^{\sum_{j=1}^{i-1}n_j}\ket{\dots n_j-1_i\dots}.
    \end{split}
\end{equation}
Isometry output Eq.~(\ref{eq:transformed qudit}) gives rise to a new class of channels we nicknamed {\em Grassmann channels}.
\begin{defi}\label{defi:grass}
    Let the $d$-dimensional Grassmann channel $\G_d:\B(\bbC^d)\mapsto\B(\bbC^{2^d-1})$ be a quantum channel defined by the action of its isometry $V_{\G_d}$ as $\G_d(\psi_{A'})= \Tr{C}\circ\ V_{\G_d}\circ\psi_{A'}$. The action of the channel on an input qudit is given by
\begin{equation}\label{eq:Grasschanneloutput}
    \G_d:\psi_{A'}\mapsto\g^{(d)}_A=\cos^{2(d-1)}{r}\bigoplus_{k=1}^{d}\tan^{2(k-1)}{r}{{d-1\choose k-1}} \chi^{(d)}_k
    =\bigoplus_{k=1}^{d}p_k\G_{d,k}(\psi_{A'}),
\end{equation}
where $p_k=\cos^{2(d-1)}{r}\tan^{2(k-1)}{r}{{d-1\choose k-1}}$. The maps $\G_{d,k}:\B(\bbC^d)\mapsto\B\big(\bbC^{{d\choose k}}\big)$ for $1\leq k\leq d$ are quantum channels constituting the Grassmann channel $\G_d$ and $\chi^{(d)}_k=\G_{d,k}(\psi_{A'})$ having $\Tr{}\big[\chi^{(d)}_k\big]=1$.
\end{defi}
\begin{rem}
The output block dimension ${d\choose k}$  coincides with the dimension of the vector spaces into which the Grassmann algebra $\Lambda(V)$ over a $d$-dimensional vector space $V$ decomposes ($\Lambda(V)=\bigoplus_{k=1}^d\Lambda^k(V)$), hence the name Grassmann channels. Note that the Grassmann algebra is also known as the exterior algebra.
\end{rem}
\begin{rem}
    The index $k$ in Eq.~(\ref{eq:Grasschanneloutput}) has two roles in $\chi^{(d)}_k$ . It labels the state but it also indicates how many fermions the basis is composed of. We thus know in which basis $\chi^{(d)}_k$ is written.
\end{rem}
\begin{rem}
    The first Grassmann channel $\G_1$ is a trivial trace map. The second Grassmann channel $\G_2$ is a channel recently playing an important role in quantum Shannon theory - a qubit erasure channel~\cite{erasure_channel,era_ass,era_code,era_priv}. We now look at $\G_2$ and $\G_3$ cases in more detail.
\end{rem}
\begin{exa}For $d=2$ we get from Eq.~(\ref{eq:fermi unitary action})
    \begin{equation}
        \begin{split}
        \ket{\Psi}_{AC}&=\cos^2{r}\exp\left[{\tan{r}(a_1^\dg c_1^\dg+a_2^\dg c_2^\dg)}\right]\ket{vac}\\
         &=\cos^2{r}\left[1+\tan{r}(a_1^\dg c_1^\dg+a_2^\dg c_2^\dg)+{\tan^2{r}\over2}(a_1^\dg c_1^\dg+a_2^\dg c_2^\dg)^2\right]\ket{vac}\\
        &=\cos^2{r}\bigg(\ket{vac}+\tan{r}\big(\ket{10}_A\ket{10}_C+\ket{01}_A\ket{01}_C\big)
         -\tan^2{r}\ket{11}_A\ket{11}_C\bigg).
        \end{split}
    \end{equation}
    Following Eq.~(\ref{eq:trans_quditc}) we obtain
    \begin{subequations}
    \begin{align}
    \ket{\Phi}_{AC}
    &={\cos{r}}\Big((\b_1\ket{10}+\b_2\ket{01})\ket{vac}
    +\tan{r}\left(\b_1\ket{11}\ket{01}-\b_2\ket{11}\ket{10}\right)\Big)\label{eq:transformed qubitb}\\
    &=\sqrt{1-p}\big(\b_1\ket{1}+\b_2\ket{2}\big)_A\ket{0}_C
    +\sqrt{p}\ket{3}_A\big(\b_1\ket{2}-\b_2\ket{1}\big)_C.\label{eq:transformed qubitc}
    \end{align}
    \end{subequations}
    The last equation is Eq.~(\ref{eq:transformed qubitb}) rewritten using a reparametrization $\sqrt{1-p}=\cos{r},\sqrt{p}=\sin{r}$ and a logical ket notation ($\ket{0},\dots,\ket{3}$) to facilitate the comparison with an isometry output for a qubit erasure channel
    \begin{equation}\label{eq:erasure_isometry}
        \ket{\Phi}^{erase}_{AC}=\sqrt{1-p}\big(\b_1\ket{1}+\b_2\ket{2}\big)_A\ket{f}_C
        +\sqrt{p}\ket{f}_A\big(\b_1\ket{1}+\b_2\ket{2}\big)_C.
    \end{equation}
    $\ket{f}$ is a flag state orthogonal to both $\ket{1}$ and $\ket{2}$. The most notable difference between Eq.~(\ref{eq:transformed qubitc}) and Eq.~(\ref{eq:erasure_isometry}) is the dimension of the output Hilbert space. For the latter the Hilbert space is three-dimensional. But if we trace over the $A$ or $C$ subsystems and compare them we immediately see that they indeed induce the same channel. The Grassmann channel $\G_2$ is just embedded in a higher-dimensional space than the corresponding erasure channel and unitarily rotated (note the difference in the second brackets of Eqs.~(\ref{eq:transformed qubitc}) and~(\ref{eq:erasure_isometry})).
\end{exa}
\begin{exa}
    For $\G_3$ we get from Eq.~(\ref{eq:trans_quditc})
    \begin{subequations}\label{eq:transformed qutrit}
    \begin{align}
    \ket{\Phi}_{AC}
    &={\cos^2{r}}\Big((\b_1\ket{100}+\b_2\ket{010}+\b_3\ket{001})\ket{vac}\label{eq:transformed qutrita}\\
    &+\tan{r}\left[(-\b_1\ket{101}-\b_2\ket{011})\ket{001}-(\b_1\ket{110}-\b_3\ket{011})\ket{010}+(\b_2\ket{110}+
    \b_3\ket{101})\ket{100}\right]\label{eq:transformed qutritb}\\
    &-\tan^2{r}\ket{111}(\b_1\ket{011}-\b_2\ket{101}+\b_3\ket{110})\Big).\label{eq:transformed qutritc}
    \end{align}
    \end{subequations}
    The trace over the $C$ subsystem gives us the output of $\G_3$
    \begin{equation}
    \G_3:\psi_{A'}\mapsto\g^{(3)}_A=p_1\chi^{(3)}_1\oplus p_2\chi^{(3)}_2\oplus p_3\chi^{(3)}_3,
    \end{equation}
    where $p_1=\cos^4{r},p_2=2\tan^2{r}\cos^4{r},p_3=\tan^4{r}\cos^4{r}$. $\chi^{(3)}_1$ is the input state itself and $\chi^{(3)}_3$ is a flag state $\ket{111}$. We are interested in the form of $\chi^{(3)}_2$
    \begin{equation}
    \chi^{(3)}_2={1\over2}\begin{pmatrix}
            |\b_2|^2+|\b_3|^2 & \b_2\bar\b_1 & -\b_3\bar\b_1 \\
            \bar\b_2\b_1 & |\b_1|^2+|\b_3|^2 & \b_3\bar\b_2 \\
            -\bar\b_3\b_1 & \bar\b_3\b_2 & |\b_1|^2+|\b_2|^2 \\
          \end{pmatrix}
    \end{equation}
    where the two-fermionic basis is ordered as $\{\ket{011},\ket{101},\ket{110}\}$. We easily verify that the Kraus operators
    \begin{equation}
    K_1={1\over\sqrt{2}}\begin{pmatrix}
          0 & -1 & 0 \\
          1 & 0 & 0 \\
          0 & 0 & 0 \\
        \end{pmatrix},
    K_2={1\over\sqrt{2}}\begin{pmatrix}
          0 & 0 & 1 \\
          0 & 0 & 0 \\
          1 & 0 & 0 \\
        \end{pmatrix},
    K_3={1\over\sqrt{2}}\begin{pmatrix}
          0 & 0 & 0 \\
          0 & 0 & 1 \\
          0 & 1 & 0 \\
        \end{pmatrix}
    \end{equation}
    yields $\chi^{(3)}_2$.
\end{exa}
\begin{rem}
     The basis in which the states and Kraus operators from the two previous examples are written are generated by the action of creation operators satisfying the canonical anticommutation relations. Yet in our case we treat the Grassmann channel outputs as if written in a multi-qubit basis and we apply the usual machinery of quantum Shannon theory (namely utilizing the tensor product structure). The following Theorem justifies this step. Note that we will illustrate it on the case of the quantum capacity. The classical capacity follows the same line of arguments.
\end{rem}
\begin{thm}\label{thm:FermiToQuditTransformation}
     The quantum capacity formula Eq.~(\ref{eq:quantcap}) remains unchanged if for an isometry $V_\K:A'\hookrightarrow AC$ the input Hilbert space is spanned by the fermionic multi-rail basis introduced in Definition~\ref{defi:multirail}.
\end{thm}
\begin{proof}
     Using the definition~of the optimized coherent information Eq.~(\ref{eq:optimcoh}) we write the relevant part of Eq.~(\ref{eq:quantcap}) for $n=2$
     \begin{equation}\label{eq:OptiCohInfo_n2}
        Q^{(1)}(\K^{\otimes 2})=\sup_{\Theta}{[H(A_1A_2)_{\tau^{\otimes 2}}-H(C_1C_2)_{\tau^{\otimes 2}}}],
     \end{equation}
     where $\Theta\equiv\Theta_{A'_1C'_1A'_2C'_2}$ is an input state. The structure of the input state is crucial. First, recall from Sec.~\ref{sec:def} that we represent isometry $V_\K$ as a unitary from Eq.~(\ref{eq:simplified_exp}) $U_{AC}:A'C'\mapsto AC$ where the input reference system $C'$ is prepared in a vacuum state. Therefore, every pair of input modes $A'_jC'_j$ contains exactly one fermion and an arbitrary input basis state $\big\{\ket{\Theta^{(k)}}_{A'_1C'_1A'_2C'_2}\big\}_{k=1}^{d^2}$ contains two fermions. The basis is clearly constructed by tensoring product basis of two Hilbert spaces $A'_1C'_1$ and $A'_2C'_2$. Note the order of the product reflected by the order of the primed subscripts. We now turn our attention to the action of $(U^{(1)}_{AC}\otimes U^{(2)}_{AC})$ on  this basis. The superscripts on the unitaries go from one to $n$ to distinguish it from $U_{A_iC_i}$ in Eq.~(\ref{eq:fermi unitary action}). We find that
     \begin{equation}\label{eq:reordered_product}
        (U^{(1)}_{AC}\otimes U^{(2)}_{AC})
        \ket{\Theta^{(k)}}_{A'_1C'_1A'_2C'_2}=\ket{\Phi^{(k_1)}}_{AC}\otimes\ket{\Phi^{(k_2)}}_{AC}
        \equiv \tau^{(k)}_{A_1C_1A_2C_2},
     \end{equation}
     where the middle equation is a product of two states from Eq.~(\ref{eq:trans_quditc}) each of them having all $\beta$'s but one equal zero ($k_1,k_2=1\dots d$). This crucial step is possible due to the fact that each pair of modes $A'_jC'_j$ contains just one fermion, the structure of $U^{(j)}_{AC}$ in Eq.~(\ref{eq:simplified_exp}) and identity~(\ref{eq:sigmaZXa}).

     Looking at Eq.~(\ref{eq:trans_quditc}) (or better at the examples in Eqs.~(\ref{eq:transformed qubitb}) and~(\ref{eq:transformed qutrit}) representing a generic situation) we see that
     \begin{equation}\label{eq:fermionicSWAP}
         \mathrm{SWAP}_{A_2C_1}\tau^{(k)}_{A_1C_1A_2C_2}\neq\tau^{(k)}_{A_1A_2C_1C_2}.
     \end{equation}
     A sign change occurs according to the number of fermions the swapped systems contain. However, when we apply the unitary product on an arbitrary input state
     \begin{equation}
        (U^{(1)}_{AC}\otimes U^{(2)}_{AC})\Bigg(\sum_{k=1}^{d^2}\alpha_k\ket{\Theta^{(k)}}_{A'_1C'_1A'_2C'_2}\Bigg)
        =\tau_{A_1C_1A_2C_2},
     \end{equation}
     where $|\alpha_k|^2=1$, and reorder the output followed by partial trace over $C_1C_2$ ($A_1A_2$) to get the (complementary) channel output, we notice another crucial property. The partial trace produces a direct sum structure where for each subspace the sign that comes from the reordering procedure is identical. The reason is that each block contains only fermionic states of a given fermionic number and its constituents come from the parts of $\tau_{A_1C_1A_2C_2}$ where the parity of the number of fermions in $C_1$ and $A_2$ remains unchanged (so the `swapping' sign is the same). So if the sign is positive (even or odd number of fermions in both subsystems) nothing happens and if the sign is negative (opposite parities in $C_1$ and $A_2$) the sign gets canceled for a given subspace. Stated differently, the phase (in this case just plus or minus one) is common for each subspace so it disappears in the density matrix formalism (blockwise) and therefore for the whole block-diagonal density matrix.

     Only little changes if we take a tensor product of $n$ unitaries in Eq.~(\ref{eq:OptiCohInfo_n2}). The dimension of the input Hilbert space will be $d^n$ and the equivalent of Eq.~(\ref{eq:reordered_product}) will hold as well as the phase argument for the resulting block-diagonal structure after taking partial trace over the $C_1\dots C_n$ or $A_1\dots A_n$ subsystem.
\end{proof}
\begin{rem}
     Consequently, provided that we use the multi-rail encoding for the isometry input, we may simply take an output state $\tau_{A(C)}$ and pretend that the basis in which it is written is an ordinary multi-qubit basis. Thus taking the tensor product in quantities defined in Eqs.~(\ref{eq:classcap}) and (\ref{eq:quantcap}) is justified by the above theorem. We may take their tensor products following the usual rules as in ordinary quantum information theory.

     Also note that in the proof of degradability of Grassmann channels in Section~\ref{sec:quantcap} we work with entire isometry outputs and so we have to fully respect the fermionic character of these states.
\end{rem}
In the proof of the next theorem we will extensively use the geometric picture from the representation theory of the $\sldc$ Lie algebras. We summarized some basic facts in Appendix~\ref{app:Liealgebras}.
\begin{thm}\label{thm:covariance}
    Let the first block of $\gamma^{(d)}_A$ in Eq.~(\ref{eq:Grasschanneloutput}) be written as
    \begin{equation}\label{eq:firstblock}
      \chi_1^{(d)}={1\over d}\Big(\openone+\sum_{\a=1}^{L}n_\a\lambda^{(1)}_\a\Big),
    \end{equation}
    where $\lambda^{(1)}_\a$ are generators of the dual representation to the fundamental representation of the $\sldc$ algebra,  $L=2d(d-1)$ and $n_\a$ are functions of $\b_i\bar\b_j$. Then the remaining blocks in  Eq.~(\ref{eq:Grasschanneloutput}) can be expanded with the same coefficients $n_\a$
    \begin{equation}\label{eq:allblocks}
        \chi_k^{(d)}={1\over{d\choose k}}\Big(\openone+\sum_{\a=1}^{L}n_\a\lambda^{(k)}_\a\Big),
    \end{equation}
    where $\lambda^{(k)}_\a$ are generators of the $k$-th completely antisymmetric representation of the $\sldc$ algebra.
\end{thm}
\begin{rem}
    The number $L$ in Eqs.~(\ref{eq:firstblock}) and~(\ref{eq:allblocks}) comes from the use of a redundant number of the $\sltw$ subalgebras as discussed in Appendix~\ref{app:Liealgebras}. Namely, there is $d(d-1)$ diagonal generators $H_n$, $d(d-1)/2$ off-diagonal generators $E_{ij}$ (recall that $i<j$) and $d(d-1)/2$ of their Hermitian conjugates~$E^\dg_{ij}$.
\end{rem}
\begin{rem}
    Since we use a linearly dependent set of algebra generators the expansion coefficients $n_\a$ are not unambiguously determined. For us, however, it is sufficient to show that at least for one specific construction (the one presented here) the coefficients can be chosen to stay preserved when switching to a higher-dimensional representation of $\sldc$.
\end{rem}
The upcoming proof will be divided into two parts. We will separately prove Eq.~(\ref{eq:allblocks}) for diagonal and off-diagonal generators. To better follow the proof it might be  helpful to watch the example of Eq.~(\ref{eq:trans_quditc}) for $d=4$. In the course of the proof there is a remark illustrating results on this case.
\begin{exa}For $d=4$ we get from Eq.~(\ref{eq:trans_quditc})
\begin{equation}\label{eq:exampled4}
    \begin{split}
        \ket{\Phi}_{AC}
        &={\cos^3{r}}\Bigg(\Big(\b_1\ket{1000}+\b_2\ket{0100}+\b_3\ket{0010}+\b_4\ket{0001}\Big)\ket{vac}\\
        &+\tan{r}\Big[
                      \Big(\b_1\ket{1001}+\b_2\ket{0101}+\b_3\ket{0011}\Big)\ket{0001}\\
        &\hspace{9mm}+\Big(\b_1\ket{1010}+\b_2\ket{0110}-\b_4\ket{0011}\Big)\ket{0010}\\
        &\hspace{9mm}+\Big(\b_1\ket{1100}-\b_3\ket{0110}-\b_4\ket{0101}\Big)\ket{0100}\\
        &\hspace{9mm}-\Big(\b_2\ket{1100}+\b_3\ket{1010}+\b_4\ket{1001}\Big)\ket{1000}\Big]\\
        &-\tan^2{r}\Big[
          \Big(\b_1\ket{1011}+\b_2\ket{0111}\Big)\ket{0011}+\Big(\b_1\ket{1101}-\b_3\ket{0111}\Big)\ket{0101}\\
        &\hspace{11mm}-\Big(\b_2\ket{1101}+\b_3\ket{1011}\Big)\ket{1001}+\Big(\b_1\ket{1110}+\b_4\ket{0111}\Big)\ket{0110}\\
        &\hspace{11mm}+\Big(-\b_2\ket{1110}+\b_4\ket{1011}\Big)\ket{1010}+\Big(\b_3\ket{1110}+\b_4\ket{1101}\Big)\ket{1100}\Big]\\
        &-\tan^3{r}\ket{1111}\Big(-\b_1\ket{0111}+\b_2\ket{1011}-\b_3\ket{1101}+\b_4\ket{1110}\Big)\Bigg).
    \end{split}
 \end{equation}
    We do not indicate the subsystems but in a product of two kets the first one is the $A$~subsystem and the second one is the $C$~subsystem.
\end{exa}

We will employ the fermionic representation of the $\sltw$ algebra~\cite{liealgebras} (see also Eqs.~(\ref{eq:sl2C}))
\begin{subequations}\label{eq:fermirep_of_sl2C}
    \begin{align}
        \begin{pmatrix}
          1 & 0 \\
          0 & -1 \\
        \end{pmatrix}
        &=a_m^\dg a_m-a_l^\dg a_l\\
        \begin{pmatrix}
          0 & 1 \\
          0 & 0 \\
        \end{pmatrix}
        &=a_m^\dg a_l\label{eq:fermirep_of_sl2Cb}\\
        \begin{pmatrix}
          0 & 0 \\
          1 & 0 \\
        \end{pmatrix}
        &=a_l^\dg a_m,
    \end{align}
\end{subequations}
where $m\not=l$ are mode labels.

\begin{proof}
    (i) Let us first find the coefficients $n_\a$ from Eq.~(\ref{eq:firstblock}) and then show that they remain preserved in  Eq.~(\ref{eq:allblocks}). Since we have
    \begin{equation}\label{eq:beta1}
         {|\b_1|^2\over d}\Big(\sum_{i=1}^da^\dg_ia_i+\sum_{j=1}^{d-1}(a_1^\dg a_1-a_{1+j}^\dg a_{1+j})\Big)
         =|\b_1|^2a_1^\dg a_1\equiv|\b_1|^2\kbr{1\dots0}{1\dots0}
    \end{equation}
    and so
    \begin{equation}\label{eq:betai}
         {|\b_i|^2\over d}\Big(\sum_{i=1}^da^\dg_ia_i
         +\sum_{j=1}^{d-1}(a_i^\dg a_i-a_{i\oplus j}^\dg a_{i\oplus j})\Big)
         =|\b_i|^2a_i^\dg a_i\equiv|\b_i|^2\kbr{0\dots1_i\dots0}{0\dots1_i\dots0}
    \end{equation}
    there is $(d-1)$ coefficients $n_\a$ that equal $|\b_i|^2$ for $i=1\dots d$. Note that $i\oplus j\equiv i+j\mod{(d+1)}$.

    For $k>1$ we first notice in Eq.~(\ref{eq:trans_quditc}) the presence of the factors $(\pm)_i$. They are irrelevant for now since we investigate the diagonal generators. The important fact is that the numerical coefficients of all $\b_i$'s are the same in the absolute value, that is, equal to one. We trace over $C$ and the diagonal part of the $A$ subsystem for a given $k$ reads
    \begin{equation}\label{eq:diagpart}
        \diag{[\chi_k^{(d)}]}={1\over{d\choose k}}\Bigg[\sum_{N_k}\sum_{i\in I}^{d-k+1}|\b_i|^2\kbr{\dots n_j+1_i\dots}{\dots n_j+1_i\dots}\Bigg].
    \end{equation}
    Each diagonal element $\sum_{i\in I}^{d-k+1}|\b_i|^2\kbr{\dots n_j+1_i\dots}{\dots n_j+1_i\dots}$ corresponding to a given $n_j$ can be factorized
    \begin{align}\label{eq:diagpart_factorized}
        &\sum_{i\in I}^{d-k+1}|\b_i|^2\kbr{\dots n_j+1_i\dots}{\dots n_j+1_i\dots}\nn\\
        &=\Bigg(\sum_{i\in I}^{d-k+1}|\b_i|^2\kbr{0\dots1_i\dots0}{0\dots1_i\dots0}\Bigg)
         \kbr{\dots n_j\dots}{\dots n_j\dots},
    \end{align}
    where the expression in the parenthesis in the second line can be rewritten using the fermionic representation of the $\sltw$ algebra from Eq.~(\ref{eq:betai}). The rest serves as a label (recall that there is ${d\choose k-1}$ orthogonal states for which $\sum n_j=k-1$). However, because of the factorization in Eq.~(\ref{eq:diagpart_factorized}) the dimension of the first expression will become $d'=d-\sum_{j=1}^d n_j$ giving us
    \begin{align}\label{eq:diagpart_factorized_rewritten}
        &\sum_{i\in I}^{d-k+1}|\b_i|^2\kbr{\dots n_j+1_i\dots}{\dots n_j+1_i\dots}\nn\\
        &=\Bigg(\sum_{i\in I}^{d-k+1}|\b_i|^2\Big(\sum_{i=1}^{d'}a^\dg_ia_i
         +\sum_{j=1}^{d'-1}(a_i^\dg a_i-a_{i\oplus j}^\dg a_{i\oplus j})\Big)\Bigg)
         \kbr{\dots n_j\dots}{\dots n_j\dots}.
    \end{align}
    So the diagonal elements for $k>1$ can also be expressed using only the $\sltw$ diagonal algebra generators. This is what we expected following the discussion in Appendix~\ref{app:Liealgebras}.

    Now we have to make sure that there is the right number of summands when constructing the direct sum $\sltw$ subalgebra representations also discussed in Appendix~\ref{app:Liealgebras} in the last paragraph. This is indeed the case.
    The orthogonal `label states' factorized out in Eq.~(\ref{eq:diagpart_factorized_rewritten}) label the subspace in which the $\sltw$ subalgebra in the parenthesis lives. We want to count how often this situation happens. Since it is a $\sltw$ subalgebra we have two spots out of $d$ in each ket occupied by two fermions (by choosing, for example, $\b_1$ and $\b_2$ or any other pair) and the remaining number of spots can be occupied by $\sum_{j=1}^d n_j$ fermions. The number of possibilities is
    $$
    \#[orthogonal\ subspaces]={d-2\choose\sum_{j=1}^d n_j}\equiv{d-2\choose k-1}.
    $$
    We thus reproduced the result found in Appendix~\ref{app:Liealgebras} for all completely antisymmetric representations of the $\sldc$ algebras.
    \begin{rem}
        We might verify the above considerations on Eq.~(\ref{eq:exampled4}). For $k=2$ (lines 2 to 5) tracing over $C$ and taking the diagonal part gives us unnormalized Eq.~(\ref{eq:diagpart}) for which $d-k+1=3$. There is ${4\choose1}=4$ states where $\sum_{j=1}^d n_j=1$ holds and they can be factorized out as in Eq.~(\ref{eq:diagpart_factorized_rewritten}). There are therefore four blocks with mutually orthogonal flags attached to them but only ${2\choose1}=2$ of them have a common pair of coefficients (e.g., $\b_1$ and $\b_2$). This is the case illustrated in Fig.~\ref{fig:GrassXXX}b. All edges of the octahedron are $\sltw$ algebras. The pairs of parallel lines form the generators of a completely antisymmetric representation of $\slfo$ in the form of a direct sum. They span a 4-dimensional space but none of the spaces is orthogonal to any other.
    \end{rem}
    (ii) Proving  Eq.~(\ref{eq:allblocks}) for off-diagonal generators is considerably simpler. Looking at the rightmost sum of Eq.~(\ref{eq:trans_quditc}) we see that every state $\ket{\dots n_j\dots}_C$ is multiplied by $\sum_j\b_j(\pm)_{j}\ket{\dots n_j+1_i\dots}_A$ containing one fermion more in the $i$-th mode compared to the $C$ subsystem. Tracing over $C$ and taking the off-diagonal part gives us expressions of the following type
    \begin{equation}\label{eq:offdiagonalgenerators}
        \sum_{l,m}\b_l\bar\b_m(\pm)_{l}(\pm)_{m}\kbr{\dots n_i+1_l\dots}{\dots n_j+1_m\dots}.
    \end{equation}
    We immediately recognize the off-diagonal matrices forming Eq.~(\ref{eq:offdiagonalgenerators}) to be proportional to the step operators pertaining to the $\sltw$ subalgebra from Eq~(\ref{eq:fermirep_of_sl2C}). The situation is only complicated by the presence of a function $(\pm)_{l}(\pm)_{m}$ responsible for the change of a sign. The function comes from the fermionic relations~Eq.~(\ref{eq:fermialgebra}). The change of sign therefore appears if we annihilate the fermion in the $l$-th mode and create it in the $m$-th mode
    $$
    (\pm)_{l}(\pm)_{m}\ket{\dots n_j+1_m\dots}_A=a_m^\dg a_l\ket{\dots n_j+1_l\dots}.
    $$
    But this is precisely the action of the operator from Eq.~(\ref{eq:fermirep_of_sl2Cb}) representing the $\sltw$ subalgebra. We again take a direct sum of ${d-2\choose k-1}$ off-diagonal generators. They form the off-diagonal generators of the $k$-th lowest completely antisymmetric representation of $\sldc$.
\end{proof}
\begin{cor}\label{cor:covariance}
    The Grassmann channels $\G_d$ are $SU(d)$ covariant.
\end{cor}
\begin{proof}
    The proof is identical to Corollary~17 for completely symmetric representations of $\sldc$ in~\cite{CMP} showing that the Unruh channel is $SU(d)$ covariant.
\end{proof}

\section{Quantum Capacity}\label{sec:quantcap}

There is an interesting symmetry between the $A$ and $C$ output subsystem captured in the following lemmas and later necessary for the proof of degradability of the Grassmann channels.
\begin{lem}\label{lem:firstswitcheslast}
    Labeling the Pauli matrices $\s_Z$ and $\s_X$ by $Z$ and $X$ and a $d$-dimensional identity by $I_d$ we first introduce an infinite product $Z\otimes I_2\otimes Z\otimes I_2\otimes Z\otimes I_2\otimes\dots$. A $d$-mode unitary operator $\Xi_d$ is defined as the first $d$ unitaries of the product. Then the following identity holds
\begin{equation}\label{eq:firstswitcheslast}
    \O_{0\to(d-1)}\circ X^{\otimes2d}\circ(\Xi_d\otimes I_d)\circ(-)^{d}\Big(\sum_{i=1}^d\b_ia_i^\dg\Big)\ket{vac}_{AC}
    =\tan^{d-1}{r}\Big(\sum_{i=1}^d\b_ic_i^\dg\Big)\ket{\vcn}_{A}\ket{\vcn}_{C},
\end{equation}
    where $\sum_{i=1}^dn_i=d-1$ and $\O_{(k-1)\to (d-k)}:\tan^{k-1}{r}\to\tan^{d-k}{r}$ is a non-physical operation solely acting on the trigonometric function in Eq.~(\ref{eq:trans_quditb}) (or Eq.~(\ref{eq:trans_quditc})). For $k=1$ we understand $\O_{0\to(d-1)}(1)=\tan^{d-1}{r}$ in Eq.~(\ref{eq:firstswitcheslast}).
\end{lem}
\begin{proof}
    We prove the identity by a direct calculation. Using properties of the fermionic algebra, the following hold:
    \begin{subequations}\label{eq:first_and_last}
    \begin{align}
        X^{\otimes2d}\circ(\Xi_d\otimes I_d)\Big(\sum_{i=1}^d\b_ia_i^\dg\ket{vac}_{AC}\Big) &=X^{\otimes2d}\circ(\Xi_d\otimes I_d)\circ \Big(\sum_{i=1}^d\b_i\ket{0\dots1_i\dots0}_A\Big)\ket{vac}_{C}\label{eq:first_and_lasta}\\
        &=(-)^iX^{\otimes2d}\circ\Big(\sum_{i=1}^d\b_i\ket{0\dots1_i\dots0}_A\Big)\ket{vac}_{C}\label{eq:first_and_lastb}\\
        &=(-)^i\Big(\sum_{i=1}^d\b_i\ket{1\dots0_i\dots1}_A\Big)\ket{1\dots1}_{C}\label{eq:first_and_lastc}.
    \end{align}
    \end{subequations}
    On the other hand we have
    \begin{equation}\label{eq:first_and_lastRHS}
        \tan^{d-1}{r}\sum_{i=1}^d\b_ic_i^\dg\ket{\vcn}_{A}\ket{\vcn}_{C}
        =\tan^{d-1}{r}\Big(\sum_{i=1}^d(-)^{d-1+i-1}\b_i\ket{1\dots0_i\dots1}_A\Big)\ket{1\dots1}_{C}.
    \end{equation}
    By multiplying Eq.~(\ref{eq:first_and_lasta}) by $(-)^{d-2}$ and exchanging the trigonometric function (represented by the action of $\O_{0\to (d-1)}$) we get the RHS of Eq.~(\ref{eq:first_and_lastRHS}) and therefore also Eq.~(\ref{eq:firstswitcheslast}).
\end{proof}
The purpose of the next lemma is to present two extensively used identities.
\begin{lem}\label{lem:on_sigma_z}
Let $a$ be a fermionic operator. Then, the following identities hold
\begin{subequations}\label{eq:sigmaZX}
    \begin{align}
          Za^\dg  &= -a^\dg Z\label{eq:sigmaZXb} \\
          Xa^\dg  &= aX\label{eq:sigmaZXc}.
    \end{align}
\end{subequations}
\end{lem}
\begin{proof}
    Directly follows from Eq.~(\ref{eq:anticomm}) once we realize that $X=a+a^\dg$ and $Z=a^\dg a-aa^\dg=2a^\dg a-1=1-2aa^\dg$.
\end{proof}
\begin{lem}\label{lem:AswitchesC}
    The following relation holds
    \begin{equation}\label{eq:AswitchesC}
        \O_{(k-1)\to(d-k)}\circ X^{\otimes2d}\circ(\Xi_d\otimes I_d)
        \circ\Big(\sum_{i=1}^d\b_ia_i^\dg\Big)\ket{\Psi}_{AC}
        =\Big(\sum_{i=1}^d\b_ic_i^\dg\Big)\ket{(\pm)_{\vcm,k,d}\Psi}_{AC},
    \end{equation}
    where $\ket{(\pm)_{\vcm,k,d}\Psi}_{AC}$ is $\ket{\Psi}_{AC}$ from Eq.~(\ref{eq:fermi unitary action}) with some of the summands having a negative sign.
\end{lem}
\begin{rem}
    The indices of the phase function will be described and the function explicitly  written.
\end{rem}
\begin{proof}
    To make the derivation smoother we will first prove a specific part of Eq.~(\ref{eq:AswitchesC}) in which $k=2$
\begin{multline}\label{eq:secswitcheslastbutone}
    \O_{1\to(d-2)}(\tan{r})\circ X^{\otimes2d}\circ(\Xi_d\otimes I_d)\circ
    \Big(\sum_{\substack{i=1 \\ i \ne j}}^d\b_ia_i^\dg\Big)a^\dg_jc^\dg_j\ket{vac}_{AC}\\
    =\tan^{d-2}{r}(-)^{j-1}\Big(\sum_{\substack{i=1 \\ i \ne j}}^d\b_ic_i^\dg\Big)\ket{\vcm^j}_{A}\ket{\vcm^j}_{C}.
\end{multline}
    where $\sum_{i=1}^d m^j_i=d-2$ and $m^j_j = 0$. Let us ignore for a while the function $\O_{1\to(d-2)}$. Using identities~(\ref{eq:sigmaZXa}), (\ref{eq:sigmaZXb}),(\ref{eq:sigmaZXc}) and Lemma~\ref{lem:firstswitcheslast} (note the missing $j$-th summand from both sides of Eq.~(\ref{eq:firstswitcheslast})) the LHS of Eq.~(\ref{eq:secswitcheslastbutone}) reads
    \begin{subequations}
        \begin{align}
        (-)^ja_jc_j\circ X^{\otimes2d}\circ(\Xi_d\otimes I_d)
        \circ\Big(\sum_{\substack{i=1 \\ i \ne j}}^d\b_ia_i^\dg\Big)\ket{vac}_{AC}
        &=(-)^j(-)^{d}a_jc_j\Big(\sum_{\substack{i=1 \\ i \ne j}}^d\b_ic_i^\dg\Big)\ket{\vcn}_{A}\ket{\vcn}_{C}\\
        &=(-)^j(-)^{d}(-)^{d-1}\Upsilon\Big(\sum_{\substack{i=1 \\ i \ne j}}^d\b_ic_i^\dg\Big)\ket{\vcm^j}_{A}\ket{\vcm^j}_{C}\label{eq:eq2},
        \end{align}
    \end{subequations}
    where $\Upsilon=(-)^{j-2}(-)^{j-2}=(-)^{j-1}(-)^{j-1}\equiv1$ depending on whether $i<j$ or $j<i$. Thus we recovered the phase on the RHS of Eq.~(\ref{eq:secswitcheslastbutone}). It is time to justify the presence of $\O_{1\to(d-2)}$ in Eq.~(\ref{eq:secswitcheslastbutone}). Eq.~(\ref{eq:trans_quditc}) indicates that the exponent of $\tan^{k-1}{r}$ coincides with the number of excitations of the $C$~subsystem and is one less than the number of excitations in the $A$~subsystem. From  reasons that will become clear later we would like the exponent of the trigonometric function $\tan^{d-2}{r}$ on the RHS to correspond to the number of excitations of the $A$~subsystem which is indeed equal to $d-2$.

    To prove  Eq.~(\ref{eq:AswitchesC}) we first rewrite Eq.~(\ref{eq:fermi unitary action}) as
    $$
        \sum_{k=1}^{d+1}\tan^{k-1}{r}(-)^{\Delta_{k-2}} \sum_{N_k} \ket{ \vec{n}^J }_A \ket{ \vec{n}^J }_C =  \sum_{k=1}^{d+1}\tan^{k-1}{r} \sum_{N_k}(a_j^{\dg} c_j^{\dg})^{n_j}\ket{ vac }_{AC} = {1\over\cos^d{r}}\ket{\Psi}_{AC},
    $$
    where we have used the notation introduced in Eq.~\eqref{eq:trans_quditc}. We have also introduced the set $J$, for a fixed $N_k$, as $J=\{ j | 1 \le j \le d,\text{ } n_j = 1\}$. Let us investigate the expression for a chosen $k$:
    \begin{subequations}\label{eq:allANDall}
        \begin{align}
            &\O_{(k-1)\to(d-k)}(\tan^{k-1}{r})\circ X^{\otimes2d}\circ(\Xi_d\otimes I_d)\circ
                \Big(\sum_{\substack{i=1 \\ i \notin J}}^d\b_ia_i^\dg\Big)\sum_{N_k}(a_j^{\dg}c_j^{\dg})^{n_j}\ket{ vac }_{AC} \label{eq:allANDalla} \\
            &=\O_{(k-1)\to(d-k)}(\tan^{k-1}{r})\circ\sum_{N_k} ((-1)^j a_j c_j)^{n_j}
                \circ X^{\otimes2d}\circ(\Xi_d\otimes I_d)\circ\Big(\sum_{\substack{i=1 \\ i \notin J}}^d\b_ia_i^\dg\Big)\ket{vac}_{AC}\\
            &=\tan^{d-k}{r}(-)^{d}\sum_{N_k}  ((-1)^j a_j c_j)^{n_j}
                \Big(\sum_{\substack{i=1 \\ i \notin J}}^d\b_ic_i^\dg\Big)\ket{\vec{n'}}_{A}\ket{\vec{n'}}_{C}\\
            &=\tan^{d-k}{r}(-)^{d}\sum_{N_k}  (-)^{j n_j}(-)^{(k-1)(d-1)}
                \Upsilon^{k-1}(-)^{\Delta_{k-2}}
                \Big(\sum_{\substack{i=1 \\ i \notin J}}^d\b_ic_i^\dg\Big)\ket{\vcm^J }_{A}\ket{\vcm^J }_{C}\\
            &=\tan^{d-k}{r}\Big(\sum_{\substack{i=1 \\ i \notin J}}^d\b_ic_i^\dg\Big)\sum_{N_k} (\pm)_{\vec{m},k,d} \ket{\vcm^J}_{A}\ket{\vcm^J}_{C}
            \label{eq:allANDalle},
        \end{align}
    \end{subequations}
    In the first equality of Eq.~\eqref{eq:allANDall} we used Eqs.~(\ref{eq:sigmaZXf}),~(\ref{eq:sigmaZXb}) and~(\ref{eq:sigmaZXc}), and in the second equality Lemma~\ref{lem:firstswitcheslast} was used with the summands corresponding to the set $J$ removed from both sides of Eq.~(\ref{eq:firstswitcheslast}). Here $\sum_in_i'=d-1$ holds. In the third equality, Eqs.~(\ref{eq:sigmaZXe}),~(\ref{eq:acorder}) and the properties of the fermionic algebra leading to $\Upsilon$ in Eq.~(\ref{eq:eq2}) were used and we introduced a vector $\ket{\vcm^{J} }$ such that $\sum_{i=1}^d m^{J}_i=d-k$ and $m^{J}_j = 0$ if $j \in J$.  Finally in the last equation we collected all phases into a single function $(\pm)_{\vec{m},k,d}$.

    To prove the lemma we realize that the LHS of Eq.~(\ref{eq:AswitchesC}) is Eq.~(\ref{eq:allANDalla}) summed over $k$
    \begin{equation}
        \cos^d{r}\sum_{k=1}^{d}\Bigg[
        \O_{(k-1)\to(d-k)}(\tan^{k-1}{r})\circ X^{\otimes2d}\circ(\Xi_d\otimes I_d)\circ\Big(\sum_{\substack{i=1 \\ i \notin J}}^d\b_ia_i^\dg\Big)\sum_{N_k} (a_j^{\dg} c_j^{\dg})^{n_j}\Bigg]
        \ket{vac}_{AC}
    \end{equation}
    (note that we sum only to $d$ since the $(d+1)$-th summand disappeared). The RHS of Eq.~(\ref{eq:AswitchesC}) is just a sum of Eq.~(\ref{eq:allANDalle}) over $k$
    \begin{equation}
        \cos^d{r}\sum_{k=1}^{d}\tan^{d-k}{r}\Big(\sum_{\substack{i=1 \\ i \notin J}}^d\b_ic_i^\dg\Big)\sum_{N_k}(\pm)_{\vcm,k,d}\ket{\vcm^J}_{A}\ket{\vcm^J}_{C}
        =\Big(\sum_{i=1}^d\b_ic_i^\dg\Big)\ket{(\pm)_{\vcm,k,d}\Psi}_{AC}
    \end{equation}
    Notice that if we exchange $(d-k)$ for $(k-1)$ in $\tan^{d-k}{r},N_k,\ket{\vcm^J}_{A}$ and  $\ket{\vcm^J}_{C}$ such that $\sum_{i=1}^d m^{J}_i=k-1$ we indeed get $\Big(\sum_{i=1}^d\limits\b_ic_i^\dg\Big)\ket{\Psi}_{AC}$ up to the phase function.
\end{proof}
\begin{rem}
    The important fact about Eq.~(\ref{eq:allANDalle}) is that the sign depends on chosen $d,k$ and $N_k$ so it is in general different for each summand over $N_k$ but common for each sum over $i$.
\end{rem}
\begin{lem}\label{lem:ZIZI}
    Following the notation of Definition~\ref{defi:grass} we label $\g^{(d)}_C=\Tr{A}\circ\,\Phi_{AC}$. Then we have
    \begin{equation}\label{eq:gamma_ACrelation}
        \O_{(k-1)\to(d-k)}\circ W\circ\g^{(d)}_A=\g^{(d)}_C,
    \end{equation}
    where $W$ is a unitary transformation.
\end{lem}
\begin{proof}
    Lemma~\ref{lem:AswitchesC} is powerful since it allows us to compare the outputs of Grassmann channels with their complementary outputs $\g^{(d)}_C$. Indeed, from the form of $\ket{\Phi}_{AC}$ we can deduce the explicit form of the block matrices from which $\g^{(d)}_A$ is composed of (see Eq.~(\ref{eq:Grasschanneloutput})). However, if we wanted to compare $\g^{(d)}_A$ with $\g^{(d)}_C$ given by unitary $U_{AC}$ in Eqs.~(\ref{eq:transformed qudit}) it would be a difficult task. The advantage of Lemma~\ref{lem:AswitchesC} is that we don't even need to know $\g^{(d)}_A$ or $\g^{(d)}_C$ explicitly to find the relation between them.

    Let us rewrite the previous lemma result in the following way
    \begin{equation}\label{eq:result of the ACswitch lemma}
      \Big(\sum_{i=1}^d\b_ia_i^\dg\Big)\ket{\Psi}_{AC}
      =\O_{(k-1)\to(d-k)}\circ(\Xi_d\otimes I_d)\circ X^{\otimes2d}\circ\Big(\sum_{i=1}^d\b_ic_i^\dg\Big)\ket{(\pm)_{\vcm,k,d}\Psi}_{AC}.
    \end{equation}
    Recall that $\O_{(k-1)\to(d-k)}\circ(\Xi_d\otimes I_d)\circ X^{\otimes2d}$ is an involution and $\O_{(k-1)\to(d-k)}$ is a scalar function. Note that $\Xi_d\otimes I_d$ commutes or anticommutes with $X^{\otimes2d}$ depending on the specific form of $\Xi_d$. We incorporate this sign change into $(\pm)_{\vcm,k,d}$. Tracing out the $A$ subsystem we have
    \begin{subequations}\label{eq:Atraced}
    \begin{align}
        \g_C^{(d)}=\Tr{A}\circ\Big(\sum_{i=1}^d\b_ia_i^\dg\Big)\ket{\Psi}_{AC}
        &=\O_{(k-1)\to(d-k)}\circ W\circ\Tr{A}\circ\Big(\sum_{i=1}^d\b_ic_i^\dg\Big)\ket{(\pm)_{\vcm,k,d}\Psi}_{AC}\label{eq:Atraceda}\\
        &=\O_{(k-1)\to(d-k)}\circ W\circ\Tr{A}\circ\Big(\sum_{i=1}^d\b_ic_i^\dg\Big)\ket{\Psi}_{AC}\label{eq:Atracedb}\\
        &=\O_{(k-1)\to(d-k)}\circ W\circ\Tr{C}\circ\Big(\sum_{i=1}^d\b_ia_i^\dg\Big)\ket{\Psi}_{AC}\label{eq:Atracedc}\\
        &=\O_{(k-1)\to(d-k)}\circ W\circ\g_A^{(d)}\label{eq:Atracedd}.
    \end{align}
    \end{subequations}
    The first equality holds since all the unitaries in Eq.~(\ref{eq:result of the ACswitch lemma}) act locally and $\O_{(k-1)\to(d-k)}$ is again a scalar function (the unitary $W$ appears as the result of tracing over unitarily-locally transformed state $\ket{\Phi}_{AC}$). The sign ambivalence in the second equality is irrelevant because tracing over $A$ means creating of a convex sum of states belonging to the $C$~subsystem so the phase disappears, the third equality is a simple permutation argument due to the symmetry of $\ket{\Psi}_{AC}$ between the modes $A$ and $C$ and in Eq.~(\ref{eq:Atracedd}) we invoke the definition of $\g_A^{(d)}$.
\end{proof}
Let us formally define the complementary Grassmann channel based on Lemma~\ref{lem:ZIZI} and Definition~\ref{defi:grass}:
\begin{defi}\label{defi:grasscompl}
    The $d$-dimensional complementary Grassmann channel $\G^c_d$ is the quantum channel defined by the isometry $V_{\G_d}$ as $\G^c_d(\psi_{A'})= \Tr{A}\circ\,V_{\G_d}\circ\psi_{A'}$. The action of the channel on an input
    qudit is given by
\begin{equation}\label{eq:Grasscomplchanneloutput}
    \G^c_d:\psi_{A'}\mapsto\g^{(d)}_C=\cos^{2(d-1)}{r}\bigoplus_{k=1}^{d}\tan^{2(d-k)}{r}{{d-1\choose k-1}}
    W\circ\chi^{(d)}_{k},
\end{equation}
where $\chi^{(d)}_k$ has been introduced in Eq.~(\ref{eq:Grasschanneloutput}).
\end{defi}
For the purpose of proving degradability, $W$ is a harmless unitary matrix independent on $k$ so we may just ignore it. But $\O_{(k-1)\to(d-k)}$ is not a completely positive map so degradability remains to be proven.
\begin{thm}\label{thm:degradability}
    All Grassmann channels $\G_d$ from Eq.~(\ref{eq:Grasschanneloutput}) are degradable for $r\in\left[0,\pi/4\right]$.
\end{thm}
\begin{rem}
    Note that the degradability of a qubit erasure channel on the whole interval is recovered for $d=2$. This corresponds to $p\in[0,1/2]$ in Eq.~(\ref{eq:transformed qubitc}) as is valid for a `standard' qubit erasure channel~Eq.~(\ref{eq:erasure_isometry})~\cite{erasure_channel}.
\end{rem}
\begin{figure}[t]
\begin{center}
      \epsfig{file=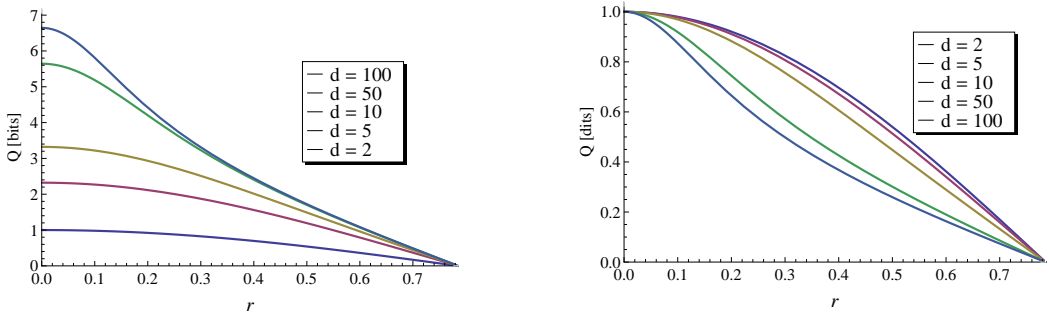,width=14.5cm}
      \caption{Quantum capacity for Grassmann channels for $d=2, 5, 10, 50$ and $100$. In the left plot the capacity is calculated using the base two logarithm and in the right plot using the base $d$ logarithm.  The curves follow the order in the legend.}
      \label{fig:quantcapGrass}
\end{center}
\end{figure}
\begin{proof}
    We will prove the theorem by a direct construction of the degrading map. Rewriting Eqs.~(\ref{eq:Grasschanneloutput}) and (\ref{eq:Grasscomplchanneloutput}) we get
    \begin{align}
      \g^{(d)}_A &= \bigoplus_{k=1}^{d} p_k\chi^{(d)}_{k}\equiv\bigoplus_{k=1}^{d} p_k\G_{d,k}(\psi_{A'}) \\
      \tilde\g^{(d)}_C
      &= W\g^{(d)}_CW^\dg = \bigoplus_{k=1}^{d} \tilde p_{k}\chi^{(d)}_{k},
    \end{align}
    where $p_k=\cos^{2(d-1)}{r}\tan^{2(k-1)}{r}{{d-1\choose k-1}}$ and $\tilde p_{k}=\cos^{2(d-1)}{r}\tan^{2(d-k)}{r}{{d-1\choose k-1}}$. Recall that for the purpose of the degrading map construction we may work with $\tilde\g^{(d)}_C$ instead of $\g^{(d)}_C$. We assume the existence of the following degrading map
    \begin{equation}\label{eq:degradingmap}
            \M:\g^{(d)}_A\mapsto q_1\g^{(d)}_A+\sum_{k=2}^{d}q_k\chi^{(d)}_{k}= \sum_{k=1}^{d}
            \tilde p_{k}\chi^{(d)}_{k}\equiv\tilde\g^{(d)}_C.
    \end{equation}
    In order for $\M$ to be a completely positive map we have to show that $0\leq q_k$ and $\sum_{k=1}^dq_k=1$ for $r\in\left[0,\pi/4\right]$. For all $d$ we get from Eq.~(\ref{eq:degradingmap}) the following set of $d$ linear equations
    \begin{equation}\label{eq:eqsset}
    \begin{split}
      q_1 p_1 &= \tilde p_1 \\
      q_1 p_2 + q_2 &= \tilde p_2 \\
       &\  \vdots \\
      q_1 p_d + q_d &= \tilde p_d.
    \end{split}
    \end{equation}
    The set is easily solvable. The first equation gives us $q_1=\tan^{2(d-1)}{r}$ which we plug into the remaining equations. We get
    \begin{equation}\label{eq:q_i}
        q_k={d-1\choose k-1}\cos^{2(d-1)}{r}\left(\tan^{2(d-k)}{r}-\tan^{2(d+k-2)}{r}\right)
    \end{equation}
    for $k=2\dots d$. Since for $r\in[0,\pi/4]$ the tangent function is monotonously increasing and $0\leq\tan{r}\leq1$ holds as well we may conclude that for $k=1\dots d$ all $q_k$ are positive. Finally, summing the left and right side of the equation set and using $\sum_{k=1}^d p_k=\sum_{k=1}^d\tilde p_k=1$ we find that $\sum_{k=1}^dq_k=1$.
\end{proof}

We might proceed to the calculation of the quantum capacity of the Grassmann channels. Due to the degradability the quantum capacity formula Eq.~(\ref{eq:quantcap}) reduces to the optimized coherent information~Eq.~(\ref{eq:optimcoh}). Furthermore, according to Theorem~\ref{thm:covariance} the Grassmann channels are covariant. Therefore, the supremum in Eq.~(\ref{eq:optimcoh}) is achieved for a maximally mixed input qudit and the quantum capacity formula reads
\begin{equation}\label{eq:quantcapGrass}
    Q(\G_d)={1\over d}\cos^{2(d-1)}{r}\sum_{k=1}^dk\binom{d}{k}\log{k}\big(\tan^{2(d-k)}{r}-\tan^{2(k-1)}{r}\big).
\end{equation}
The plots for various $d$ can be found in Fig.~\ref{fig:quantcapGrass} where we have plotted the capacities using both the base two (on the left) and base $d$~logarithm (on the right). The reason for the presence of the base $d$~logarithm is to compare the capacity in a more fair way.

\section{Classical Capacity}\label{sec:classcap}

We first present a simple generalization of the characterization theorem derived in~\cite{WolfEisert} (Theorem~1).
\begin{thm}[\cite{WolfEisert}]\label{thm:wolfeisert}
    The quantum channel $T:\B\big(\mathbb{C}^d\big) \rightarrow \B\big(\mathbb{C}^{d'}\big)$ is of the form
    $$
    T(\rho)=\dfrac{I_{d'}- m M(\rho)}{d'-m},
    $$
    where $M$ is a positive, linear and trace-preserving map such that there exists a state $\rho_0$ where $mM(\rho_0)$ is a projection of rank $m$ if and only if the $\a$-R\'enyi minimal output entropy $H^\a_{min}(T)=\min_\rho{H^\a(T(\rho))}$ is $\a$-independent.
\end{thm}
\begin{rem}
The (almost trivial) generalization lies in setting $d'\not=d$.
\end{rem}
We now show that the blocks from which all the Grassmann channels are composed fulfill the required conditions.
\begin{lem}
Every block of the qudit Grassmann channel $\G_d: \B\big(\mathbb{C}^d\big) \rightarrow \B\big(\mathbb{C}^{2^d-1}\big)$ has the following form:
\begin{align}
\label{eq:wolfform}
\G_{d,k}(\rho) = \dfrac{I_{d'}- m M_{d,k}(\rho)}{d'-m}
\end{align}
where $d'={d\choose k}$ is the dimension of the output space. Moreover, for all input pure states $\rho_0$,  $m M_{d,k}(\rho_0)$ is a projection of rank $m$
\end{lem}
\begin{proof}
Consider the $k$-th block of the qudit Grassmann channel with input of the form $\kbr{1}{1}$. Then the output state will have the following form:
\begin{align}
{\phi^{(k)}_A}&= \dfrac{1}{ {d-1 \choose k-1} } \sum_{s_i} \kbr{1\,s_2 \hdots s_d}{1\,s_2 \hdots s_d}_A
\end{align}
where we are using the convention $\sum_{i=2}^d s_i = k-1$ and normalized the block. Now suppose we let $m={ d-1 \choose k }$ and we know the dimension of the output space for the $k$-th block is equal to ${ d \choose k }$, then $d'-m = { d \choose k } - { d-1 \choose k } = { d-1 \choose k-1 }$. Therefore,
\begin{align}
I_{d'}-(d'-m)\G_{d,k}(\kbr{1}{1}) &= I_{d'} -  \sum_{s_i} \kbr{1\,s_2 \hdots s_d}{1\,s_2 \hdots s_d}\\
\label{eq:rankargument}
&= \sum_{t_i} \kbr{0\,t_2 \hdots t_d}{0\,t_2 \hdots t_d}.
\end{align}
with the convention now being $\sum_{i=2}^d t_i = k$. The matrix in~Eq.~(\ref{eq:rankargument}) has rank $m={d-1\choose k}$. Thus with an input of the form $\rho_0=\kbr{1}{1}$, we have found a matrix $M_{d,k}(\rho_0)$ that satisfies the condition in~Eq.~(\ref{eq:wolfform}).

Now, in order to generalize this result to arbitrary pure inputs, one can use the $SU(d)$ covariance of the Grassmann channel presented in Corollary~\ref{cor:covariance}. Assume $R_{d,r}$ to be an $r$-dimensional unitary representation of $SU(d)$. Then due to covariance the following holds:
\begin{align}
    \G_{d,k}(\rho) &= \G_{d,k}( R_{d,d}\circ\rho_0)= R_{d,d'}\circ\G_{d,k} (\rho_0).
\end{align}
Thus the following chain of equalities hold:
\begin{subequations}
    \begin{align}
        R_{d,d'}\circ\big(mM_{d,k}(\rho_0)\big)&=R_{d,d'}\circ\big((I_{d'} - (d'-m) \G_{d,k}(\rho_0)\big)\\
        &= I_{d'} - (d'-m) R_{d,d'}\circ \G_{d,k}(\rho_0)  \\
        &= I_{d'} - (d'-m) \G_{d,k}(R_{d,d}\circ\rho_0) \\
        &= mM_{d,k}(R_{d,d}\circ\rho_0).
    \end{align}
\end{subequations}
\begin{figure}[h]
\begin{center}
      \epsfig{file=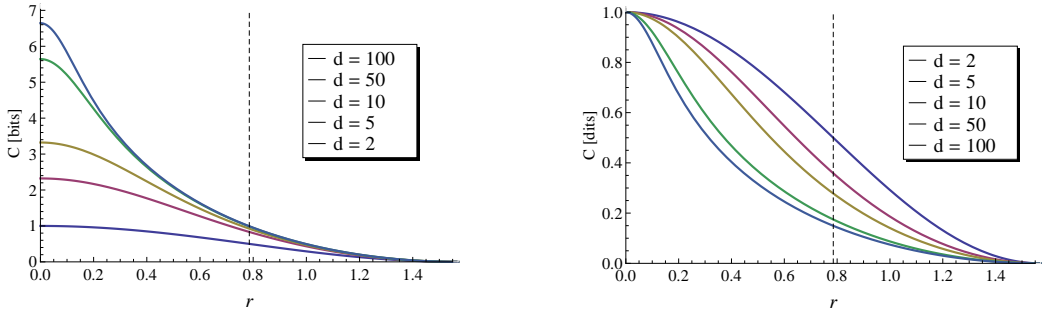,width=14.5cm}
      \caption{Classical capacity for Grassmann channels for $d=2, 5, 10, 50$ and $100$. The vertical dashed lines at $r=\pi/4$ label the point where the quantum capacity in Fig.~\ref{fig:quantcapGrass} equals zero. In the left plot the capacity is calculated using the base two logarithm and in the right plot using the base $d$ logarithm.  The curves follow the order in the legend.}
      \label{fig:classcapGrass}
\end{center}
\end{figure}
The last equality follows from defining $mM_{d,k}(R_{d,d}\circ\rho_0)$ in this way, since it has rank $m$ since it is just equal to a rotated state of rank $m$. Thus every block, after normalization, has the form of~Eq.~(\ref{eq:wolfform}).
\end{proof}
\begin{cor}
Grassmann channels have a direct sum form Eq.~(\ref{eq:Grasschanneloutput}). Hence from Lemma~3 in~\cite{classcap_additivity_clone} it follows that all Grassmann channels have the Holevo capacity additive. From Eqs.~(\ref{eq:classcap}) and (\ref{eq:holcap}) we get
\begin{equation}
    C(\G_d)=\sup_{\{p(x),\s_{x,A'}\}}{\bigg[H(A)_{\G_d(\s_{XA'})}
    -\sum_xp(x)H(A)_{\G_d(\s_{x,A'})}\bigg]}.
\end{equation}
Consequently, exploiting the result from~\cite{CovCh,hadamard} for covariant channels, the classical capacity is given by
\begin{align}\label{eq:classcapGrass}
     C(\G_d)&=H(\{p_k\})+\sum_{k=1}^dp_k\log{d\choose k} - H(A)_{\G_d(\rho_0)}\nn\\
     &=\log{d}-\cos^{2(d-1)}{r}\sum_{k=1}^d\tan^{2(k-1)}{r}{d-1\choose k-1}\log{k}.
\end{align}
\end{cor}

\section{Physical implications}\label{sec:consequences}

\subsection*{Bosons, fermions and non-inertial observers}

It is often claimed that there is a fundamental difference between the entanglement behavior of maximally entangled states built upon fermions and bosons in a relativistic setting~\cite{unruh-fermi}. Namely, provided that a uniformly accelerated observer has one half of an initially maximally entangled state it is found that in the infinite acceleration limit the bosonic entanglement disappears while the fermionic entanglement partially persists. Based on this na\"ive approach it is concluded that in the infinite limit entangled fermionic states might be useful for various quantum-informational protocols where a certain amount of shared entanglement is usually needed.

However, our capacity results suggest something different at least for the purpose of quantum communication between an inertial and noninertial observer. For the purpose of sending quantum messages the ultimate measure of a channel's capability to transmit the information is its quantum capacity which has nothing to do with any particular entanglement measure. We found that there is no qualitative difference between the behavior of the quantum capacity for the qudit Grassmann channel (fermions) and its bosonic equivalent introduced in Ref.~\cite{CMP} where we study the bosonic version of transformation Eq.~(\ref{eq:exponential}). This led to the definition of the qudit Unruh channels as the bosonic counterpart of the qudit Grassmann channel. One of the resolved problems is the quantum capacity of the qudit Unruh channels which we can compare with the Grassmann channels. As a result we find that quantum capacities for both channels converge to zero as we approach the infinite acceleration limit. So from the viewpoint of quantum Shannon theory entangled resources based on bosons or fermions are equally useful. The only difference is in the capacity value for a finite acceleration. To fairly compare the capacities we rewrite Eq.~(\ref{eq:quantcapGrass}) as a function of $w=\tan^2{r}$. Using $\cos^2{r}=1/(1+\tan^2{r})$ we get
\begin{subequations}
\begin{align}\label{eq:equantcap_Grass_rewritten}
    Q(\G_d)&={1\over d}\left({1\over1+w}\right)^{d-1}\sum_{k=1}^dk{d\choose k}\log{k}\big(w^{d-k}-w^{k-1}\big)\\
    &=\dfrac{1}{(1+w)^{d+1}} \sum_{k=0}^{d-1} w^k {d-1 \choose k} \log{\dfrac{d-k}{k+1}},
\end{align}
\end{subequations}
where the details of derivation leading to the second line can be found in Appendix~\ref{app:capratio}.
\begin{figure}[t]
\begin{center}
      \epsfig{file=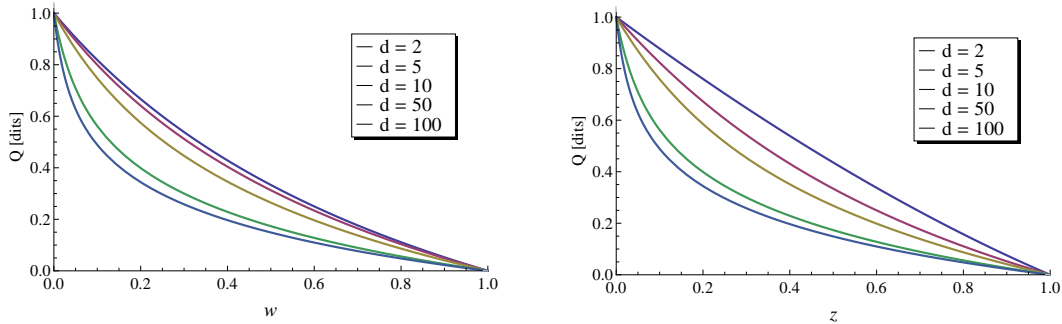,width=14.5cm}
      \caption{Quantum capacity for the Grassmann channel (on the left) and the Unruh channel (on the right) taken from~\cite{CMP}. The parameters $w=\tan^2{r}$ and $z=\tanh^2{r}$ are chosen to compare the capacities for the same proper acceleration. The capacity of the Unruh channel is in general higher but importantly for the infinite acceleration limit ($w=z=1$) both capacities are equal zero. The curves follow the order in the legend.}
      \label{fig:grassVSunruhQcap}
\end{center}
\end{figure}
This expression can be directly compared to the one we get for the qudit Unruh channel $\U_d$~\cite{CMP}
\begin{equation}\label{eq:equantcap_Unruh}
    Q(\U_d)={1\over d}(1-z)^{d+1}\sum_{k=1}^\infty
    k{d+k-1\choose k}\log{d+k-1\over k}z^{k-1}.
\end{equation}
The Unruh channel outperforms the Grassmann channel for low-dimensional inputs as we can see in Fig.~\ref{fig:grassVSunruhQcap}. In the infinite acceleration limit $z=w\to1$ the ratio of the quantum capacities remains finite for all $d$
\begin{equation}\label{eq:capratio_def}
    r_d=\lim_{z\to1}\dfrac{Q(\mathcal{G}_d)}{Q(\mathcal{U}_d)}
    = \dfrac{d \ln{d}}{d-1} \dfrac{1}{2^{d-1}}
    \sum_{k=0}^{\lfloor \frac{d-1}{2} \rfloor}(d-1-2k) {d-1 \choose k} \log{\dfrac{d-k}{k+1}}.
\end{equation}
The details of derivation are presented in Appendix~\ref{app:capratio}. We can therefore conclude that the capacities asymptotically behave in the same way. We have not found an analytical form for the sum in $r_d$ but it is clear that $r_2=\ln{2}\leq r_d<r_{d+1}$ for all $d$. Moreover, numerical simulations suggest that $r_\infty=1$.

As a closing comment note that due to the local similarities between the Rindler and Schwarzschild spacetime (explicitly spelled out, for example, in~\cite{unruh-schwarz}) the similar conclusion regarding the quantum capacities also holds in the black hole scenario.

\subsection*{Grassmann channels and transpose-depolarizing channels}

Let us define the family of qudit transpose-depolarizing channels~\cite{TransDepol}.
\begin{defi}
    Let $\T:\B(\bbC^d)\to\B(\bbC^d)$ be a map defined as
    \begin{equation}\label{eq:transdepCh}
        \T(\s)=t\bar\s+(1-t)\openone/d
    \end{equation}
    acting on normalized density matrices $\s$. The bar denotes complex conjugation in a given basis and the map is a quantum channel in the following interval of $t$:
    \begin{equation}\label{eq:t_parameter}
        -{1\over d-1}\leq t\leq{1\over d+1}.
    \end{equation}
\end{defi}
\begin{rem}
    The transpose-depolarizing channel for $t= -1/(d-1)$ is known as the $d$-dimensional Werner-Holevo channel~\cite{WH}.
\end{rem}
\begin{lem}
    The complementary channels of the $\G_{d,2}$ channels presented in Definition~\ref{defi:grass} are the $d$-dimensional Werner-Holevo channels.
\end{lem}
\begin{proof}
    We take the corresponding part of the isometry output Eq.~(\ref{eq:trans_quditc}) for $k=2$ and rewrite it as
\begin{equation}
    \begin{split}
        \Phi^{(2)}_{AC}&=\cos^{d-1}{r} \tan{r} \sum_{j=1}^d\sum_{\genfrac{}{}{0pt}{}{i=1}{i\not=j}}^{d-1}(\pm)_{i}\b_i
        \ket{0\dots1_j\dots1_i\dots0}_A\ket{0\dots 1_j\dots0}_C\\
        &\equiv\sum_{j=1}^d\sum_{\genfrac{}{}{0pt}{}{i=1}{i\not=j}}^{d-1}(\pm)_{i}\b_i
        \ket{[ji]}_A\ket{j}_C.
    \end{split}
\end{equation}
    Note that $\ket{[ji]}=-\ket{[ij]}$ so we can further write (leaving out the irrelevant trigonometric functions)
    \begin{equation}
        \tilde{\Phi}^{(2)}_{AC}
        =\sum_{j=1}^d\sum_{i>j}^{d-1}\ket{[ji]}_A\big(-\b_i\ket{j}+\b_j\ket{i}\big)_C.
    \end{equation}
    Considering an input pure state $\ket{\psi}_{A'}=\sum_{l=1}^d\b_l\ket{l}_{A'}$ the isometry of $\G_{d,2}$ reads
    \begin{equation}\label{eq:isometryk2}
        V_{\G_{d,2}}=\sum_{j=1}^d\sum_{i>j}^{d-1}
        \Big(-\d_{il}\ket{[ji]}_A\kbr{j}{l}_{CA'}+\d_{jl}\ket{[ji]}_A\kbr{i}{l}_{CA'}\Big).
    \end{equation}
    Tracing over $A$ followed by normalizing by $\sqrt{1/{d-1\choose k-1}}=\sqrt{1/(d-1)}$ and changing the overall sign leads to ${d\choose2}$ Kraus operators of the form
    \begin{equation}\label{eq:Gd2Kraus}
        \K_{ij}={\sqrt{1\over d-1}}\big(\kbr{j}{i}-\kbr{i}{j}\big),
    \end{equation}
    where $1\leq i<j\leq d$. These are well known as the Kraus operators for the $d$-dimensional Werner-Holevo channels~\cite{WH}.
\end{proof}
This result brings us two interesting points. As mentioned earlier,  in~\cite{capregion_had} we study the capacity region of the bosonic version of the transformation from~Eq.~(\ref{eq:exponential}) known as the qudit Unruh channel~\cite{CMP}. One of the results is the characterization of the complementary channels of the Unruh channel. Their structure is also block-diagonal as in the present case and the first nontrivial complementary block for each $d$ is the transpose-depolarizing channel~Eq.~(\ref{eq:transdepCh}) for $t=1/(d+1)$. This is a peculiar observation and raises a number of questions. First of all, why for fermions we get the transpose-depolarizing channel from one end of the allowed interval~Eq.~(\ref{eq:t_parameter}) and for bosons from the other end? Also, does some sort of intermediate statistics interpolating between bosons and fermions correspond to the whole interval? One of the obvious possibilities are anyons whose appearance is not limited just to the two-dimensional world~\cite{haldane}.

The identification of the $d$-dimensional Werner-Holevo channel also implies that the Grassmann channels do not belong to the class of Hadamard channels~\cite{capregion_had,capregion_qudits} -- the channels whose complementary channel is entanglement-breaking. The reason is that the Werner-Holevo channels are known not to be entanglement-breaking. The transpose-depolarizing channels are entanglement-breaking for $-1/(d^2-1)\leq t$~\cite{TransDepol}. Henceforth, at least one of the complementary blocks of all Grassmann channels has negative partial transpose.

\section{Conclusions}

We have introduced a new class of quantum channels to the group with computable classical and quantum capacities, the Grassmann channels. Such channels are rare in quantum Shannon theory since the calculation of the classical and quantum capacities requires an optimization over an infinite number of channel uses. The Grassmann channels' isometric extension is physically well motivated and stems from the Bogoliubov transformation which preserve the canonical anticommutation relations corresponding to the Fermi-Dirac statistics.

In order to determine the quantum capacity of the Grassmann channels, we have shown that these channels are degradable and have explicitly calculated the degrading map. Combining this with the result that the Grassmann channels are covariant we were able to give a closed form for the quantum capacity. A different technique was used in order to calculate the classical capacity of the set of Grassmann channels. We showed that each block, in the block diagonal form of the matrix, had a particular form~\cite{WolfEisert} which enabled the calculation of their minimum output entropy. Exploiting this result allowed us to calculate the classical capacity of the Grassmann channels.

To appreciate the capacity results from the physical point of view we compare the Grassmann channels with the Unruh channels studied elsewhere~\cite{CMP,capregion_qudits}. They share a close analog to the Grassmann channels in the sense that their isometric extension appears formally identical. The difference is that for the Unruh channels the isometry is built upon the operators obeying the canonical commutation relations (relevant to the Bose-Einstein statistics) inducing a completely different class of channels. In particular, the Unruh channels belong to the set of Hadamard channels - the channels whose complementary channels are entanglement breaking. The set of Grassmann channels do not fall in this class. This is the first example to our knowledge of a set of channels that do not belong to the Hadamard class that have a computable and at the same time nonzero classical and quantum capacity.

The main physical consequences also come from the comparison between the fermionic and bosonic case. One of the discussed physical motivations for investigating this type of channel is that both the fermionic and bosonic Bogoliubov transformation occurs in the study of particle production in uniformly accelerated frames. Fermionic entanglement between two parties who originally share a maximally entangled state exists to an extent even in the limit of infinite acceleration of one of the participants. Contrary to the fermionic case, bosonic entanglement vanishes in the infinite acceleration limit. Based on this observation it is believed that there is a difference between these types of resources. The result of our work suggests that at least for quantum communication purposes there is no difference whatsoever. The Unruh channel does demonstrate a greater quantum capacity than the Grassmann channel when the acceleration parameter is small, however in the infinite acceleration limit, both quantum capacities tend to zero.

There exists another connection between the Unruh and Grassmann channels. Both channels are direct sums of other quantum channels and so are their complementary channels. The first non-trivial block of the complementary channel for a given dimension $d$ belongs to the family of qudit transpose-depolarizing channels which is a single-parameter family of quantum channels. Interestingly, this holds both for the Grassmann channels and Unruh channels. The difference is that the complementary block of the Unruh channel corresponds to the transpose-depolarizing channel with the parameter at the upper limit of the allowed parameter interval, while the first nontrivial complementary block of the Grassmann channel corresponds to the lower limit of the allowed interval (the Werner-Holevo channel). This immediately raises the question: Why does the fermionic case occupy one end of the interval while the bosonic case occupy the other? Perhaps and even more interestingly, could there be intermediate statistics model (like anyons, for example) that would explain intermediate values of the interval? Another direction in which future research could be done is to consider the case of coupled fermions with a more active role of the spin variable and one can ask how the additional spin variable  alters the  Grassmann channels and whether the classical and quantum capacity is still calculable. Finally, the implications of the results here obtained to quantum information protocols inspired by Cooper pairs in solid state or atomic physics scenarios deserve an independent detailed study.

\acknowledgments
K.~B. acknowledges support from the Office of Naval Research (grant No. N000140811249) and appreciates comments made by Patrick Hayden and Omar Fawzi. T.~J.~would like to acknowledge the support of NSERC through the USRA Award Program and the Alexander Graham Bell Canada Graduate Scholarship.

\appendix
\section{Geometric picture of the $\sldc$ Lie algebra representations}\label{app:Liealgebras}

In the $\sldc$ case there may be several mutually commuting operators.  A maximal linearly-independent commuting set of operators of a (semi-simple) Lie algebra is called a Cartan subalgebra.  Once a Cartan subalgebra has been chosen we can use the common eigenvectors to label the basis vectors of an irreducible
representation.

\begin{defi}
An $r$-tuple $\mathbf{\alpha} =(\a_1,\ldots,\a_r)$ of complex numbers is
called a root if: (i) not all the $\a_i$ are zero, (ii) there is an
element $E_\a$ of $\sldc$ such that
\begin{equation}\label{eq:cartanweyl}
[H_i,E_\a]=\a_iE_\a.
\end{equation}
\end{defi}

The set $\{H_i,E_\a\}$ is called the Cartan-Weyl basis~\cite{liealgebras}. The Cartan subalgebra of $\sldc$ has dimension $r=d-1$; we say that the rank of the Lie algebra is $r$.  We write $\mathbf{H}=(H_1,\ldots,H_r)$ for the Cartan subalgebra generated by the elements $\{H_1,\ldots,H_r\}$ of the Lie algebra; these elements are assumed to be independent. Among all the roots there is a class of special roots called simple roots.
\begin{defi}
A root is called a simple root if it cannot be written as a linear combination of other positive roots.
\end{defi}
\begin{defi}
  If $\rho$ (where $\rho:\sldc\to GL(V)$ for some $V$) is a representation   of $\sldc$ then a $r$-tuple $\mathbf{\mu} =(\mu_1,\ldots,\mu_r)$ of complex numbers is a weight for $\rho$ if there is a nonzero vector   $\psi\in V$ such that $\psi$ is an eigenvector of each $H_i$ with eigenvalue $\mu_i$.
\end{defi}
If $\mathbf{\mu}$ is a weight and $\psi$ a weight vector for $\rho$ and $\mathbf{\alpha}$ is a root then
\begin{equation}\label{eq:weighteq}
\rho(H_i)\rho(E)\psi = (\mu_i + \a_i) \rho(E)\psi.
\end{equation}

In short, $E$ changes all the eigenvalues of the Cartan operators and it
creates a new weight vector (or kills the weight vector).  The root is a
vector in the weight space that points in the direction in which the
weights are changing. The $E$ operators are called shift or raising and lowering operators.
Roughly speaking, the positive roots correspond to raising operators while the negative roots to lowering operators. We classify the irreducible representations by the \emph{highest possible value of the weight}.
For general semi-simple Lie algebras we do exactly the same thing once we have a suitable order on the weights in order to define the right notion of highest weight.

A special example of the Cartan-Weyl basis is the Chevalley-Serre basis~\cite{liealgebras}. Two aspects make this basis special. (i) The step operators are associated to simple roots and (ii) the normalization is chosen such that the roots are integers. Unless explicitly stated we work in this basis due to its accessible geometric interpretation.

To give the operators a geometric interpretation we will work with a specific matrix representation of the $\sldc$ algebra. Let us define $E_{ij}$, where $1\leq i\not=j\leq d$, as the matrix having one where the $i$-th row and the $j$-th column intersect and the rest are zeros. Furthermore we define a diagonal matrix $H_{i}$ in which the $i$-th diagonal entry is 1, the $(i+1)$-th diagonal entry is $-1$ and the rest are zeros. If we assume $j=i+1$ the following set $\{H_i,E_{ij},E^\dg_{ij}\}$ forms the Chevalley-Serre basis for $\sldc$. More explicitly, for $d=2$ we get
\begin{subequations}\label{eq:sl2C}
\begin{align}
    H_1&=\begin{pmatrix}
           1 & 0 \\
           0 & -1 \\
         \end{pmatrix}\\
    E_{12}&=\begin{pmatrix}
             0 & 1 \\
             0 & 0 \\
           \end{pmatrix}\\
    E^\dg_{12}&=\begin{pmatrix}
          0 & 0 \\
          1 & 0 \\
    \end{pmatrix}.
\end{align}
\end{subequations}
Therefore in this basis we have
\begin{subequations}\label{eq:sl2cinChevalley}
\begin{align}
    [H_i,E_{ij}]&=2E_{ij}\\
    [H_i,E^\dg_{ij}]&=-2E^\dg_{ij}.
\end{align}
\end{subequations}
The simple roots are elements of a vector space dual to the one spanned by elements of the Cartan subalgebra.

This structure opens the door to an insightful geometric picture in terms of the so-called root space diagram. The dual space will be called the space of roots. For the $\sldc$ Lie algebra the space is $(d-1)$-dimensional and the simple root vectors defined as an $r$-tuple of simple roots form a non-orthogonal basis. Let us call the basis spanning the space of roots the root basis. We are aware of the overuse of the expressions root and root vectors. The terminology is not stabilized and differ in various textbooks. Also, the root space diagrams are sometimes called weight diagrams. It can be shown that each consecutive simple root vectors subtend the angle $2\pi/3$. We rewrite Eq.~(\ref{eq:cartanweyl}) as
\begin{align}
[H_i,E_{ij}]&=(\mu^{(i)}_i-\mu^{(j)}_i)E_{ij},
\end{align}
where $\mu^{(j)}_i$ are called fundamental weights. The word fundamental reflects the fact that we are dealing only with simple roots. We explicitly rewrite Eq.~(\ref{eq:weighteq}) as the spectral decomposition
$$
H_{i}=\sum_{j=1}^d\mu^{(j)}_i\kbr{\psi_j}{\psi_j},
$$
where $j$ labels the $j$-th component of a $d$-component spinor $\ket{\psi}$. The important role played by the (fundamental) weights is that they are coordinates of the eigenvectors in the space of roots.
\begin{figure}[t]
\begin{center}
    \resizebox{12cm}{4cm}{\includegraphics{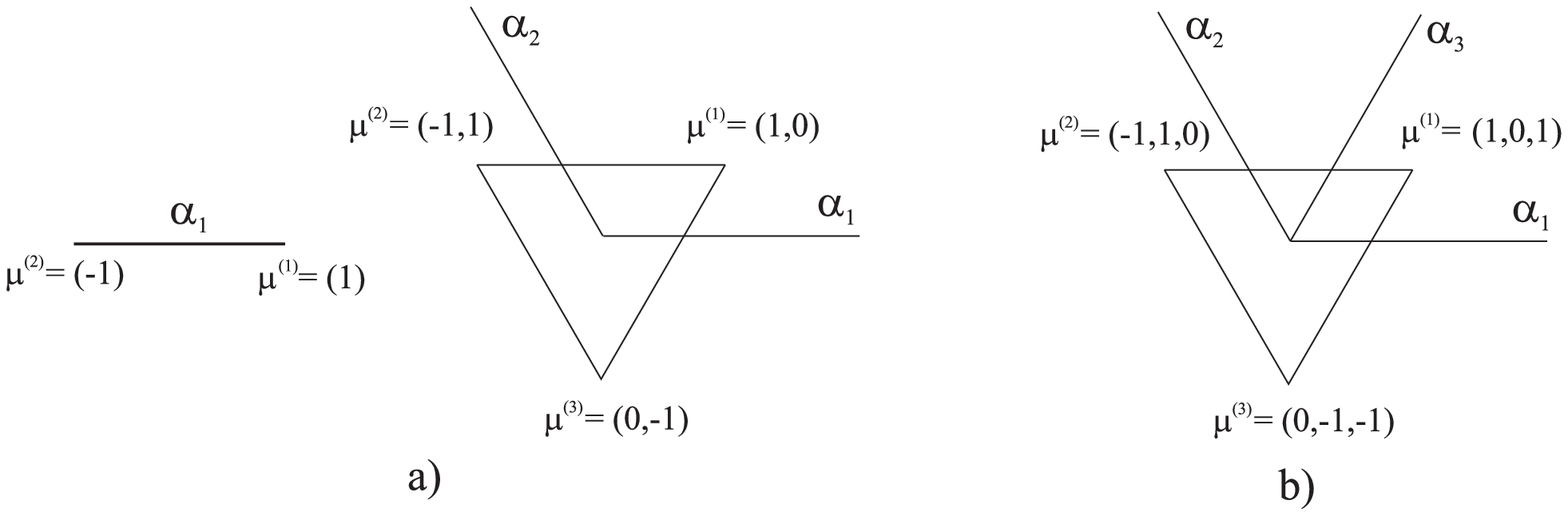}}
    \caption{(a) We illustrate the spaces of roots for $\sltw$ and $\slthr$. The basis vectors are simple roots and are indicated by $\a_i$. For the weights (vertices) we explicitly write down their coordinates in this basis.  We can read off the generators of the Cartan subalgebra from the weights. (b) It is sometimes helpful to introduce an overcomplete basis. The position of $\a_3$ reflects the relation $H_3=H_1+H_2$ for $H_3$ in Eq.~(\ref{eq:H3}) and it is not a simple root.}
    \label{fig:GrassX}
\end{center}
\end{figure}
Fig.~\ref{fig:GrassX}a illustrates the situation for $d=2$ and $d=3$. Hence for each fundamental representation of $\sldc$ there is $d$ points each representing an eigenvector ${\psi_j}$.

What about the role played by the shift operators? They have a precise geometric interpretation as well. All points of the $\sldc$ fundamental representations are interconnected. The operator responsible for a transition from site $\ket{\psi_i}$ to $\ket{\psi_j}$ is the operator $E_{ij}$ or $E^\dg_{ij}$ for the opposite direction. It also follows from Eqs.~(\ref{eq:sl2cinChevalley}) that each segment connecting two neighboring spinors has length two.

The fundamental representations of $\sldc$ algebra contain $r=d-1$ independent $\sltw$ subalgebras each satisfying Eqs.~(\ref{eq:sl2C}). However, since the root space diagram is a complete graph there are in total ${d\choose2}$ linearly dependent $\sltw$ subalgebras corresponding to the number of edges. The diagonal generators of the `additional' $\sltw$ subalgebras are constructed similarly to the $H_i$'s above. The only difference is that 1 and $-1$ on the diagonal are separated by one or more zeros. As an example ($d=3$), the remaining diagonal generator is
\begin{equation}\label{eq:H3}
H_3=\begin{pmatrix}
    1 & 0 & 0 \\
    0 & 0 & 0 \\
    0 & 0 & -1 \\
  \end{pmatrix}.
\end{equation}
Therefore, the rank of the weight vectors equals three and it correspond to introducing an overcomplete root basis, see Fig.~\ref{fig:GrassX}b. Note that $H_3$ does not correspond to a simple root. Indeed, the axis $\a_3$ in Fig.~\ref{fig:GrassX}b can be obtained by a linear combination of $\a_1$ and $\a_2$ which are both positive root vectors.

\subsection*{Completely antisymmetric representations of the $\sldc$ algebra}

For the purpose of this article we are interested in particular higher-dimensional representations of $\sldc$ - the completely antisymmetric representations. The lowest-dimensional antisymmetric representation is formed as the dual of the fundamental representation
\begin{equation}\label{eq:antisymspinors}
\psi^{j_1}=\ve^{j_1\dots j_d}\psi_{j_2}\dots\psi_{j_d},
\end{equation}
where $\ve^{j_1\dots j_d}$ is a completely antisymmetric tensor and $1\leq j_1,\dots,j_d\leq d$. The dual representation (of the fundamental representation only!) coincides with its complex conjugate representation. We get the dual representation from the fundamental representation by $G_\sldc\mapsto-\bar{G}_\sldc$ where $G_\sldc$ are all generators of $\sldc$. The eigenvalues remain the same and the weight vectors just change the sign. Hence, the root space diagrams are the same and they are just, vaguely speaking, pointing in the opposite direction. The interpretation of the edges and points follows the fundamental case. The difference lies in the fact that the antisymmetric representations of $\sldc$ are carried by $d$-component antisymmetrized spinors.

Higher-dimensional completely antisymmetric representations of the $\sldc$ algebra are formed in an intuitive way. 
First of all, the dimension of the spaces of roots remains the same as well as the number of algebra generators. The generators clearly satisfy the same commutation relations since it is just a different representation of the same $\sldc$ algebra. The root diagram `is grown' in the direction of roots but this process cannot go on forever. The spinors carrying the higher-dimensional representation are completely antisymmetrized and so there are only ${d\choose k}$ states in the $k$-th completely antisymmetric representation of $\sldc$ where $1\leq k\leq d$ (formally, we should also include the trivial representation, $k=0$).
\begin{figure}[h]
\begin{center}
      \epsfig{file=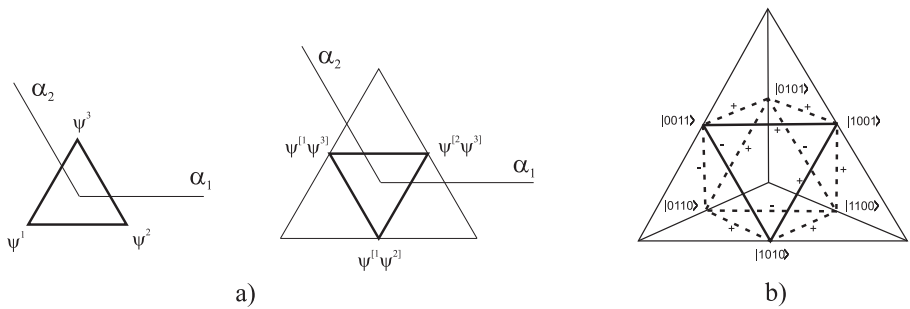,width=12cm}
      \caption{(a) The left plot is the dual representation to the fundamental representation of $\slthr$ where $\psi^i$ are completely antisymmetric spinors from Eq.~(\ref{eq:antisymspinors}). The middle plot is the second-lowest ($k=2$) completely antisymmetric representation of $\slthr$. Due to the antisymmetrization procedure (indicated by the square brackets) the vertices do not correspond to any state. (b) The second lowest ($k=2$) completely antisymmetric representation of $\slfo$. The antisymmetrization procedure leaves us with a six-dimensional space (the inner octahedron) with the spanning basis indicated. The spinors are written in the Fock basis following Definition~\ref{defi:fermstates}. The sign attached to an edge indicates the sign induced by the (fermionic) $\sltw$ shift operator when applied on spinors the edge connects.}
      \label{fig:GrassXXX}
\end{center}
\end{figure}
Fig.~\ref{fig:GrassXXX} illustrates the growth of these representations for $d=3$ and $d=4$. The connection to the fermionic states brought in Definition~\ref{defi:fermstates} is straightforward: The state $\ket{F}=\ket{n_1\dots n_d}$ for which $\sum_{i=1}^dn_i=k$ holds is an antisymmetric spinor carrying the $k$-th~completely antisymmetric representation of $\sldc$.

We find the matrix form of the generators of all higher-dimensional completely antisymmetric $\sldc$ algebra representations. Because of the complete antisymmetry condition there are only segments connecting two neighboring points. Therefore only the lowest-dimensional representation of the $\sltw$ subalgebra appears in the construction. However, because for a given edge there might be more segments parallel to it, the $\sldc$ subalgebra generators are formed by a direct sum of the $\sltw$ subalgebras. The reason for a direct sum is that they, by construction, act on mutually orthogonal subspaces. Finally, for every $d$ there is only $r=d-1$ linearly independent directions (or, said otherwise, only $r=d-1$ independent sets of parallel lines) so the number of generators equals the number of independent $\sltw$ subalgebras and they manifestly satisfy the commutation relations for $\sldc$. Note that it does NOT mean that the completely antisymmetric representations of $\sldc$ are direct sum representations. As a matter of fact, they are irreducible. The direct sum subalgebras we have  created do not themselves span mutually orthogonal subspaces. For illustration see Fig.~\ref{fig:GrassXXX}b where any pair of parallel segments `share' spinors with some other pair of parallel segments.

The only ambiguity lies in the sign of the shift operators. We can see from Eq.~(\ref{eq:sl2C}) that switching their sign does not spoil the commutation relations. Let us illustrate the ambiguity on an example from Fig.~\ref{fig:GrassXXX}b. The transition from state $\ket{0011}$ to $\ket{0101}$ is provided by $a_2^\dg a_3\ket{0011}=\ket{0101}$ meanwhile to get from $\ket{0011}$ to $\ket{0110}$ we acquire a minus sign $a_2^\dg a_4\ket{0011}=-\ket{0110}$. Not accidentally, the operators $a_i^\dg a_j,a_j^\dg a_i$ are the fermionic representation of the shift operators of the $\sltw$ algebra~\cite{liealgebras}. They play an important role in the proof of Theorem~\ref{thm:covariance} where the fermionic representation is properly introduced.

So far we talked about a direct sum of several $\sltw$ subalgebras but for the sake of proof of Theorem~\ref{thm:covariance} we need to specify how many summands there actually is. This transforms into the question how many different $\sltw$ subalgebras corresponding to a chosen direction exist. Let the direction be chosen by the step operators $\{a_2^\dg a_1,a_1^\dg a_2\}$. Then for a given $d$ and $k$ there is ${d-2\choose k-1}$ parallel segments in this particular direction for the $k$-th lowest completely antisymmetric representation of $\sldc$. We get this number by a combinatorial argument: We have a spinor with $d$ positions where the first two slots are occupied. For the lowest-dimensional representation ($k=1$) the rest is occupied by zeros and therefore for the $k$-th lowest dimensional representation there is $(k-1)$ ones to occupy the remaining $(d-2)$ free spaces. By a simple permutation argument we can see that this holds for any of the $(d-1)$ linearly independent directions. For the example in Fig.~\ref{fig:GrassXXX}b we indeed get two parallel segments for each of the three independent directions.

\section{The capacity ratio derivation}\label{app:capratio}

On can show that Eq.~(\ref{eq:equantcap_Grass_rewritten}) simplifies in the following way:
\begin{align}
Q (\mathcal{G}_d) &= \dfrac{1}{d(1+w)^{d+1}} \sum_{k=1}^{d} (w^{d-k} - w^{k-1}) k {d \choose k} \log{k}\\
&= \dfrac{1}{d(1+w)^{d+1}} \left( \sum_{l=0}^{d-1} w^l (d-l) {d \choose d-l} \log{(d-l)} - \sum_{k=1}^{d} w^{k-1} k {d \choose k} \log{k} \right)\\
&= \dfrac{1}{(1+w)^{d+1}} \sum_{k=0}^{d-1} w^k {d-1 \choose k} \log{\dfrac{d-k}{k+1}}\\
&= \dfrac{1}{(1+w)^{d-1}} \sum_{k=0}^{\lfloor \frac{d-1}{2} \rfloor} (w^k - w^{d-1-k}) {d-1 \choose k} \log{\dfrac{d-k}{k+1}} \\
& = \dfrac{1-w}{(1+w)^{d-1}} \sum_{k=0}^{\lfloor \frac{d-1}{2} \rfloor} w^k (1 + w + \cdots w^{d-2-2k}) {d-1 \choose k} \log{\dfrac{d-k}{k+1}} .
\end{align}

Now consider the Unruh channel quantum capacity from Eq.~(\ref{eq:equantcap_Unruh}),
\begin{align}\label{eq:Unruh_Qcapacity}
    Q (\mathcal{U}_d) &= \dfrac{(1-z)^{d+1}}{d} \sum_{k=0}^{\infty} z^{k-1} k {d+k-1 \choose k} \log{\dfrac{d+k-1}{k}}.
\end{align}
Now as $z \rightarrow 1$ the dominant terms in this summation are the terms such that $k \gg 1$. Thus in the limit $z \rightarrow 1$ we can approximate the $\log$ using the following:
\begin{align}\label{eq:approximation}
\log{\dfrac{d+k-1}{k}} = \log{\Big(1+\dfrac{d-1}{k}\Big)} \simeq \dfrac{1}{\ln{d}}\dfrac{d-1}{k}.
\end{align}
Let us recall that $\log$ is the logarithm base $d$ and $\ln$ denotes the natural logarithm. Eq.~(\ref{eq:Unruh_Qcapacity}) becomes:
\begin{align}
Q' (\mathcal{U}_d) &= \dfrac{d-1}{d} \dfrac{(1-z)^{d+1}}{\ln{d}} \sum_{k=1}^{\infty} z^{k-1} {d+k-1 \choose k}\label{eq:approximateQCap}\\
&= \dfrac{d-1}{d} \dfrac{(1-z)^{d+1}}{\ln{d}} \dfrac{1 - (1-z)^d}{z(1-z)^d} \\
&= \dfrac{d-1}{d \ln{d}} \dfrac{1-z}{z} ( 1 - (1-z)^d).
\end{align}
Finally, for $w=z\to1$ the ratio reads
\begin{align}
\lim_{z\to1}\dfrac{Q(\mathcal{G}_d)}{Q(\mathcal{U}_d)}
&= \lim_{z\to1}\dfrac{d \ln{d}}{d-1} \dfrac{z}{(1+z)^{d-1}( 1 - (1-z)^d)} \sum_{k=0}^{\lfloor \frac{d-1}{2} \rfloor} z^k (1 + z + \cdots z^{d-2-2k}) {d-1 \choose k} \log{\dfrac{d-k}{k+1}} \\
&= \dfrac{d \ln{d}}{d-1} \dfrac{1}{2^{d-1}} \sum_{k=0}^{\lfloor \frac{d-1}{2} \rfloor}(d-1-2k) {d-1 \choose k} \log{\dfrac{d-k}{k+1}} .
\end{align}

The following lemma makes sure that the approximation in Eq.~(\ref{eq:approximation}) is valid.
\begin{lem}\label{lem:approx}
    The difference $|Q(\mathcal{U}_d) - Q'(\mathcal{U}_d)|$ approaches zero at a rate $\mathcal{O} \left( (1-z)^2 \right)$ for $z\to1$.
\end{lem}
\begin{proof}
    We can split Eq.~(\ref{eq:Unruh_Qcapacity}) to approximate the logarithm,
    \begin{align}
    Q (\mathcal{U}_d) &= \dfrac{(1-z)^{d+1}}{d} \left[ \sum_{k=1}^{d-1} z^{k-1} k {d+k-1 \choose k} \log{\dfrac{d+k-1}{k}} + \sum_{k=d}^{\infty} z^{k-1} k {d+k-1 \choose k} \log{\dfrac{d+k-1}{k}} \right].
    \end{align}
    For $k \ge d$ we can Taylor expand the logarithm,
    \begin{align}
    \log{ \dfrac{d+k-1}{k} } = \dfrac{1}{\ln{d}} \ln{\left( 1 + \dfrac{d-1}{k} \right)} = \dfrac{1}{\ln{d}} \left[ \dfrac{d-1}{k} - \dfrac{1}{2}\left(\dfrac{d-1}{k}\right)^2 + \dfrac{1}{3}\left(\dfrac{d-1}{k}\right)^3 - \cdots \right].
    \end{align}
    Since this is a converging alternating series, we can bound the following quantity for $k \ge d$,
    \begin{align}\label{eq:differencebound}
    \left| \dfrac{d-1}{k \ln{d} } -  \log{ \left(1+\dfrac{d-1}{k}\right) } \right| = \dfrac{1}{\ln{d}} \left[ \dfrac{1}{2}\left(\dfrac{d-1}{k}\right)^2
    -\dfrac{1}{3}\left(\dfrac{d-1}{k}\right)^3 + \cdots \right] \le \dfrac{1}{2 \ln{d} }\left(\dfrac{d-1}{k}\right)^2.
    \end{align}
    Consider the following approximation of the quantum capacity (rewritten Eq.~(\ref{eq:approximateQCap}))
    \begin{align}
    Q' (\mathcal{U}_d) &= \dfrac{(1-z)^{d+1}}{d} \left[ \sum_{k=1}^{d-1} z^{k-1} k {d+k-1 \choose k} \dfrac{d-1}{k \ln{d}} + \sum_{k=d}^{\infty} z^{k-1} k {d+k-1 \choose k} \dfrac{d-1}{k \ln{d} } \right].
    \end{align}
    The difference between $Q$ and $Q'$ can be bounded as follows:
    \begin{align}
    |Q(\mathcal{U}_d) - Q'(\mathcal{U}_d)| &\le \dfrac{(1-z)^{d+1}}{d} \sum_{k=1}^{d-1} z^{k-1} k {d+k-1 \choose k} \left| \log{\dfrac{d+k-1}{k} } - \dfrac{d-1}{k \ln{d}}\right| \nonumber \\
    & \qquad +  \dfrac{(1-z)^{d+1}}{d} \sum_{k=d}^{\infty} z^{k-1} k {d+k-1 \choose k} \left| \log{\dfrac{d+k-1}{k} } -\dfrac{d-1}{k \ln{d}} \right| \\
    & \le \dfrac{(1-z)^{d+1}}{d\ln{d}} \sum_{k=1}^{d-1} z^{k-1} k {d+k-1 \choose k} \dfrac{d-1}{k} \nonumber \\
    & \qquad +  \dfrac{(1-z)^{d+1}}{d\ln{d}} \sum_{k=d}^{\infty} z^{k-1} k {d+k-1 \choose k} \dfrac{1}{2} \left( \dfrac{d-1}{k} \right)^2 \label{eq:differencebound_applied}\\
    & =  \dfrac{(d-1)(1-z)^{d+1}}{d\ln{d}} \sum_{k=1}^{d-1} z^{k-1}  {d+k-1 \choose k} \nonumber \\
    & \qquad +  \dfrac{(d-1)^2 (1-z)^{d+1}}{2d\ln{d}} \sum_{k=d}^{\infty} z^{k-1}  {d+k-1 \choose k}   \dfrac{1}{k}\\
    & \le  \dfrac{(d-1)(1-z)^{d+1}}{d\ln{d}} \sum_{k=0}^{d-1}  {d+k-1 \choose k} \nonumber \\
    & \qquad +  \dfrac{(d-1)^2 (1-z)^{d+1}}{2d\ln{d}} \sum_{k=d}^{\infty} z^{k-1}  {d+k-1 \choose k}  \dfrac{1}{k} \dfrac{k}{d+k-1} \dfrac{2d-1}{d} \label{eq:multiply} \\
    &=  \dfrac{(d-1)(1-z)^{d+1}}{d\ln{d}} {2d-1\choose d} \nonumber\\
    &\qquad+  \dfrac{(d-1)^2 (2d-1) (1-z)^{d+1}}{2d^2\ln{d}} \sum_{k=d}^{\infty} z^{k-1}  {d+k-2 \choose k} \\
    &\le  \dfrac{(d-1)(1-z)^{d+1}}{d\ln{d}}  {2d-1\choose d} \nonumber\\
    & \qquad+  \dfrac{(d-1)^2 (2d-1) (1-z)^{d+1}}{2d^2\ln{d}} \sum_{k=1}^{\infty} z^{k-1}  {d+k-2 \choose k} \\
    &= \dfrac{(d-1)(1-z)^{d+1}}{d\ln{d}}  {2d-1\choose d}\nonumber\\
    &\qquad +  \dfrac{(d-1)^2 (2d-1) (1-z)^{d+1}}{2d^2\ln{d}}  \left( \dfrac{1}{z(1-z)^{d-1}} \right) \\
    &= \mathcal{O} \left( (1-z)^{d+1} \right) + \mathcal{O} \left( (1-z)^2 \right),
    \end{align}
    where in Eq.~(\ref{eq:differencebound_applied}) we applied Eq.~(\ref{eq:differencebound}) and in Eq.~\eqref{eq:multiply} we have used that $\dfrac{k}{d+k-1} \ge \dfrac{d}{2d-1} $ for $k \ge d$.
\end{proof}


\begin{thebibliography}{99}
\bibitem{early}A. S. Holevo, Problems on Information Transmission {\bf9}, 177 (1973).
\bibitem{todo}B. Schumacher and M. A. Nielsen,  Physical Review A {\bf 54}, 2629 (1996).
    D. P. DiVincenzo, P. W. Shor and J. Smolin. Physical Review A {\bf 57}, 830 (1998).
    H. Barnum, E. Knill and M. A. Nielsen, IEEE Transactions on Information Theory {\bf46}, 1317 (2000).
    P. W. Shor, Communications on Mathematical Physics {\bf 246} 453 (2004).
    I. Devetak, IEEE Transactions on Information Theory, {\bf51} 44 (2005).
    P. Hayden and A. Winter, Communications on Mathematical Physics {\bf 284}, 263 (2008).
    I. Devetak, A. Harrow and A. Winter, IEEE Transactions on Information Theory {\bf54}, 4587 (2008).
    M. B. Hastings, Nature Physics {\bf5}, 255 (2009).
    K. Li, A. Winter, X. Zou and G. Guo, Physical Review Letters {\bf103}, 120501 (2009).
\bibitem{era_priv}G. Smith and J. Yard, Science {\bf321}, 1812 (2008).
\bibitem{classcap}A. S. Holevo, IEEE Transactions on Information Theory {\bf 269}, 44 (1998).
    B. Schumacher and M. D. Westmoreland, Physical Review A {\bf 56}, 131 (1997).
\bibitem{LSD}P. W. Shor, Lecture notes, MSRI workshop on quantum computation, November 2002.
    I. Devetak, IEEE Transactions on Information Theory {\bf 51}, 44 (2005).
    S. Lloyd, Physical Review A {\bf55}, 1613 (1997).
\bibitem{privacy}I. Devetak, IEEE Transactions on Information Theory {\bf51}, 44 (2005).
    G. Smith, J. Renes and J. Smolin, Physical Review Letters {\bf100}, 170502 (2008).
\bibitem{ent_ass}C. H. Bennett, P. W. Shor, J. A. Smolin and A. V. Thapliyal, IEEE Transactions
    on Information Theory {\bf48}, 2637 (2002).
    C. H. Bennett, P. W. Shor, J. A. Smolin and A. V. Thapliyal, Physical Review Letters {\bf83}, 3081 (1999).
\bibitem{sym_ass}G. Smith, Physical Review A {\bf78} 022306, (2008).
\bibitem{capregion}P. W. Shor, Quantum Information, Statistics, Probability p.~144-152, Rinton Press.
    M.-H. Hsieh and M. Wilde, IEEE Transactions on Information Theory {\bf5}, 4682 (2010).
    I. Devetak, A. W. Harrow and A. Winter, IEEE Transactions on Information Theory {\bf54}, 4587 (2008).
\bibitem{capregion_had}K.~Br\'adler, D. Touchette, P. Hayden and M. Wilde, Physical Review A {\bf 81} 062312, (2010).
\bibitem{capregion_qudits}T. Jochym-O'Connor, K.~Br\'adler and M. Wilde, arXiv:1103.0286.
\bibitem{classcap_additivity}C. King, IEEE Transactions on Information Theory {\bf49}, 221 (2003).
    C. King, Journal of Mathematical Physics {\bf43}, 4641 (2002).
    N. Datta and M. B. Ruskai, Journal of Physics A {\bf38}, 9785 (2005).
    M. Fukuda. Journal of Physics A: Mathematical and General {\bf38}, L753 (2005).
    R. Alicki and M. Fannes, Open Systems \& Information Dynamics {\bf11}, 339 (2004).
    B. Rosgen, Journal of Mathematical Physics {\bf49} 102107, (2008).
\bibitem{classcap_additivity_clone}K. Br\'adler, arXiv:0903.1638, accepted for publication in IEEE Transactions on Information Theory.
\bibitem{hadamard}C. King, K. Matsumoto, M. Nathanson and M. B. Ruskai, Markov Processes and Related Fields {\bf13}, 391 (2007).
\bibitem{entbreak}P. W. Shor, Journal of Mathematical Physics {\bf43}, 4334 (2002).
\bibitem{TransDepol}M. Fannes, B. Haegeman, M. Mosonyi and D. Vanpeteghem, arXiv:quant-ph/0410195.
    N. Datta, A. S. Holevo and Y. Suhov, International Journal on Quantum Information {\bf4}, 85 (2006).
\bibitem{WolfEisert}M. M. Wolf and J. Eisert, New Journal of Physics {\bf7}, 93 (2005).
\bibitem{erasure_channel}C. H. Bennett, D. P. DiVincenzo and J. Smolin, Physical Review Letters {\bf78}, 3217 (1997).
\bibitem{era_ass}D. Leung, J. Lim and P. W. Shor, Physical Review Letters {\bf103}, 240505 (2009).
\bibitem{era_code}M. Grassl, T. Beth and T. Pellizzari, Physical Review A {\bf56}, 33 (1997).
\bibitem{disentangle}S. M. Barnett and P. M. Radmore, Methods in theoretical quantum optics (Oxford University Press, USA, 1997).
\bibitem{CMP}K.~Br\'adler, P. Hayden and P. Panangaden, arXiv:1007.0997.
\bibitem{JHEP}K.~Br\'adler, P. Hayden and P. Panangaden, Journal of High Energy Physics 08 074 (2009).
\bibitem{ConjDeg}K. Br\'adler, N. Dutil, P. Hayden and A. Muhammad, Journal of Mathematical Physics
    {\bf51}, 072201 (2010).
\bibitem{degchan}I. Devetak and P. W. Shor, Communications in Mathematical Physics {\bf 256}, 287 (2005).
\bibitem{degchan_struc}T. S. Cubitt, M. B. Ruskai and G. Smith, Journal of Mathematical Physics {\bf49}, 102104
    (2008).
\bibitem{giovannetti}V. Giovannetti and R. Fazio, Physical Review A {\bf71}, 032314 (2005).
\bibitem{fermions}M. Keyl and D.-M. Schlingemann, Journal of Mathematical Physics {\bf51}, 023522
    (2010). M.-C. Ba\~nuls, J. I. Cirac, and M. M. Wolf, Physical Review A {\bf76}, 022311 (2007). H. Moriya, Journal of Physics A: Mathematical and General {\bf39}, 3753 (2006).  P. Caban, K. Podlaski, J. Rembieli\'nski, K. A. Smoli\'nski and Z. Walczak, Journal of Physics A: Mathematical and General {\bf38}, L79 (2005).
\bibitem{grass_cap}F. Caruso and V. Giovannetti,  Physical Review A {\bf76}, 042331 (2007).
\bibitem{cahglaub}K. E. Cahill and R. J. Glauber, Physical Review A {\bf59}, 1538 (1999).
\bibitem{CovCh}A. S. Holevo, arXiv:quant-ph/0212025.
\bibitem{Schrieffer}N. N. Bogoliubov, Journal of Experimental and Theoretical Physics {\bf34}, 58 (1958) [English translation: Soviet Physics, Journal of Experimental and Theoretical Physics {\bf34}, 41 (1958)].
    J. R. Schrieffer, Theory of Superconductivity (Benjamin, Massachusetts, 1964).
\bibitem{condmat}H. L. St\"ormer, D. C. Tsui and A. C. Gossard, Review of Modern Physics {\bf71}, S298 (1999).
    G. Murthy and R. Shankar, Review of Modern Physics {\bf75}, 1101 (2003).
    M. Imada, A. Fujimori and Y. Tokura, Review of Modern Physics {\bf70}, 1039 (1998).
    E. Dagotto, Review of Modern Physics {\bf66}, 763 (1994). A. Georges, G.
    Kotliar, W. Krauth and M. J. Rozenberg, Review of Modern Physics {\bf68}, 13 (1996).
    J. Park et. al.,
    Nature {\bf417}, 722 (2002).
\bibitem{atoms}C. Honerkamp and W. Hofstetter, Physical Review Letters {\bf92}, 170403 (2004).
    G. B. Partridge, W. Li, R. I. Kamar, Y. Liao and R. G. Hulet, Science {\bf 311}, 506 (2006).
    S. Riedl et. al.,
    Physical Review A, {\bf78}, 053609 (2008).
\bibitem{rocio}R. J\'auregui, M. Torres and S. Hacyan, Physical Review D {\bf43}, 3979 (1991).
    P. Langlois, Physical Review D {\bf70}, 104008 (2004).
\bibitem{curved} R. Verch, Communications in Mathematical Physics {\bf223}, 261 (2001).
\bibitem{unruh76}W. G. Unruh, Physical Review D {\bf14}, 870 (1976).
\bibitem{unruh-fermi}J. Le\'on and E. Mart\'in-Mart\'inez, Physical Review A {\bf80}, 012314 (2009).
\bibitem{unruh-schwarz}E. Mart\'in-Mart\'inez, L. Garay and J. Le\'on, Physical Review D {\bf82}, 064006 (2010).
\bibitem{liealgebras} R. Gilmore, Lie groups, Lie algebras, and some of their applications (Wiley, New~York, 1974).
    J.~Fuchs and C. Schweigert, Symmetries, Lie algebras and representations (Cambridge University Press, 2003).
\bibitem{WH}R. F. Werner and A. S. Holevo, Journal of Mathematical Physics {\bf43}, 4353 (2002).
\bibitem{haldane}F. D. M. Haldane, Physical Review Letters {\bf67}, 937 (1991).


\end{thebibliography}
\end{document}